\documentclass{article}
\usepackage[utf8]{inputenc}
\usepackage{graphicx,xcolor,wrapfig,subcaption,fullpage}
\usepackage{amsmath,amsthm,amssymb}
\usepackage{bm,authblk}
\usepackage{wasysym} 
\usepackage[lined,boxed,commentsnumbered,ruled,vlined]{algorithm2e}
\usepackage{hyperref}
\theoremstyle{plain}
\newtheorem{theorem}{Theorem}[section]
\newtheorem{proposition}[theorem]{Proposition}
\newtheorem{lemma}[theorem]{Lemma}

\theoremstyle{definition}
\newtheorem{definition}[theorem]{Definition}

\theoremstyle{remark}

\usepackage{todonotes}

\makeatletter
\newtheorem*{rep@theorem}{\rep@title}
\newcommand{\newreptheorem}[2]{
\newenvironment{rep#1}[1]{
 \def\rep@title{#2 \ref{##1}}
 \begin{rep@theorem}}
 {\end{rep@theorem}}}
\makeatother
\newreptheorem{theorem}{Theorem}
\newreptheorem{proposition}{Proposition}
\newcommand{\mathsep}{,~}
\newcommand{\st}{\ensuremath{~\middle|~}}
\newcommand{\card}[1]{\left\lvert{#1}\right\rvert}
\newcommand{\absv}[1]{\card{#1}}
\newcommand{\norm}[2]{\left\lVert{#1}\right\rVert_{#2}}

\newcommand{\pdf}[1]{{\bf p} \left[ #1 \right]}

\newcommand{\argmin}[2]{\underset{#1}{\text{argmin}} \set{#2}}
\newcommand{\set}[1]{\left\lbrace #1 \right\rbrace}

\newcommand{\setR}{\mathbb R}

\newcommand{\optcontributorscore}{\theta^{\textbf{raw}}}

\newcommand{\globalscore}{\rho^{\textbf{sss}}}

\newcommand{\contributor}{u}
\newcommand{\contributorbis}{v}

\newcommand{\Contributor}{U}
\newcommand{\EmailVerifiedContributor}{U_\checkmark}

\newcommand{\contributorscore}{\theta}

\newcommand{\zeroshiftquantile}{q^{zero}_{shift}}
\newcommand{\sssdevquantile}{q^{\bf sss}_{dev}}
\newcommand{\solidago}{\textsc{Solidago}}

\newcommand{\contributorloss}[1]{\mathcal L_{#1}^{score}}
\newcommand{\mehestan}{\textsc{Mehestan}}
\newcommand{\clip}{\textsc{Clip}}
\newcommand{\clippedmean}{\textsc{ClipMean}}
\newcommand{\brmean}{\textsc{BrMean}}
\newcommand{\mean}{\textsc{Mean}}
\newcommand{\qrmed}{\textsc{QrMed}}
\newcommand{\qrmedloss}[1]{\mathcal L_{\qrmed{}_{#1}}}

\newcommand{\qrdev}{\textsc{QrDev}}
\newcommand{\qruncertainty}{\textsc{QrUnc}}

\newcommand{\qrqtl}{\textsc{QrQtl}}
\newcommand{\eigentrust}{\textsc{EigenTrust}}
\newcommand{\pagerank}{\textsc{PageRank}}
\newcommand{\byztrust}{\textsc{LipschiTrust}}
\newcommand{\huber}{\textsc{Huber}}

\newcommand{\private}{\textsc{private}}
\newcommand{\accept}{\checkmark}
\newcommand{\reject}{\times}
\newcommand{\Alternative}{E}
\newcommand{\alternative}{e}
\newcommand{\alternativebis}{f}
\newcommand{\alternativeter}{g}
\newcommand{\alternativequater}{h}
\newcommand{\AlternativePair}{EF}

\newcommand{\comparison}{r}
\newcommand{\maxComparison}{R_{max}}
\newcommand{\normalizedcomparison}{\tilde r}

\newcommand{\bigO}{\mathcal O}
\newcommand{\iteration}{t}
\newcommand{\nbIterations}{T}

\newcommand{\votingright}{w}

\newcommand{\lipschitz}{L}

\newcommand{\uncertainty}{\Delta}
\newcommand{\defaultdeviation}{\uncertainty_{default}}
\newcommand{\scalinguncertainty}{\uncertainty \scaling}
\newcommand{\translationuncertainty}{\uncertainty \translation}

\newcommand{\optcontributorscoreuncertainty}{\uncertainty \optcontributorscore}

\newcommand{\scaling}{s}
\newcommand{\translation}{\tau}
\newcommand{\FractionOfNonVerified}{f_{\times}}
\newcommand{\defaultnonverifiedvotingright}{\votingright_{\times, default}^{total}}
\newcommand{\clippedmeancenter}{\textsc{center}}
\newcommand{\clippedmeanradius}{\textsc{radius}}

\newcommand{\pretrust}[1]{\textsc{trust}^{pre}_{#1}}
\newcommand{\interimtrust}[2]{\textsc{tr}^{#1}_{#2}}

\newcommand{\maxovertrust}{\overline{\textsc{OverTrust}}}
\newcommand{\vouchdecay}{\beta}
\newcommand{\trust}[1]{\textsc{trust}_{#1}}
\newcommand{\vouch}{V}
\newcommand{\Vouch}{\mathcal V}
\newcommand{\sinkvouch}{\vouch_{\accept}^{sink}}

\newcommand{\privacypenalty}{\votingright_{\private{}}^{penalty}}
\newcommand{\minvotingright}[1]{\votingright_{min, #1}^{public}}
\newcommand{\overtrust}{\textsc{overTrust}}

\newcommand{\byztrusterror}{\varepsilon_{\byztrust{}}}

\newcommand{\contributorpriorweight}{\alpha_{prior}^{user}}

\newcommand{\privacypenaltyvalue}{0.5}
\newcommand{\defaultnonverifiedvotingrightvalue}{2}
\newcommand{\FractionOfNonVerifiedValue}{0.1}
\newcommand{\eigentrusterrorvalue}{10^{-8}}

\usepackage{listings}
\usepackage{xcolor}
\colorlet{punct}{red!60!black}
\definecolor{background}{HTML}{EEEEEE}
\definecolor{delim}{RGB}{20,105,176}
\colorlet{numb}{magenta!60!black}
\lstdefinelanguage{json}{
    basicstyle=\small\ttfamily,
    showstringspaces=false,
    breaklines=true,
    frame=lines,
    backgroundcolor=\color{background},
    literate=
     *{0}{{{\color{numb}0}}}{1}
      {1}{{{\color{numb}1}}}{1}
      {2}{{{\color{numb}2}}}{1}
      {3}{{{\color{numb}3}}}{1}
      {4}{{{\color{numb}4}}}{1}
      {5}{{{\color{numb}5}}}{1}
      {6}{{{\color{numb}6}}}{1}
      {7}{{{\color{numb}7}}}{1}
      {8}{{{\color{numb}8}}}{1}
      {9}{{{\color{numb}9}}}{1}
      {:}{{{\color{punct}{:}}}}{1}
      {,}{{{\color{punct}{,}}}}{1}
      {\{}{{{\color{delim}{\{}}}}{1}
      {\}}{{{\color{delim}{\}}}}}{1}
      {[}{{{\color{delim}{[}}}}{1}
      {]}{{{\color{delim}{]}}}}{1},
}

\title{Solidago: A Modular Collaborative Scoring Pipeline}

\author[1,2]{Lê Nguyên Hoang}
\author[1]{Romain Beylerian}
\author[3]{Bérangère Colbois}
\author[1]{Julien Fageot}
\author[1]{Louis Faucon}
\author[1]{Aidan Jungo}
\author[1]{Alain Le Noac'h}
\author[1]{Adrien Matissart}
\author[3]{Oscar Villemaud}

\affil[1]{Tournesol Association, Switzerland}
\affil[2]{Calicarpa, Switzerland}
\affil[3]{EPFL, Switzerland}

\begin{document}
\maketitle

\begin{abstract}
This paper presents \solidago{}, 
an end-to-end modular pipeline
to allow any community of users 
to collaboratively score any number of entities.
\solidago{} proposes a six-module decomposition. 
First, it uses pretrust and peer-to-peer vouches to assign trust scores to users.
Second, based on participation, trust scores are turned into voting rights per user per entity.
Third, for each user, a preference model is learned from the user's evaluation data.
Fourth, users' models are put on a similar scale.
Fifth, these models are securely aggregated.
Sixth, models are post-processed to yield human-readable global scores.
We also propose default implementations of the six modules, 
including a novel trust propagation algorithm, 
and adaptations of state-of-the-art scaling and aggregation solutions.
Our pipeline has been successfully deployed on the open-source platform \url{tournesol.app}.
We thereby lay an appealing foundation for the collaborative,
effective, scalable, fair, interpretable and secure scoring of any set of entities.
\end{abstract}
\section{Introduction}
\label{sec:introduction}
In 2021, Twitter initiated a vote-based system,
now known as the \emph{Community Notes},
whereby users can not only propose the addition of a contextual note to the posts that are published on the platform,
but also weigh in on which of other users' contextual notes ought be to shown.
The governance of the \emph{Community Notes} is transparent and fully community-driven,
without any particular right for Twitter to intervene~\cite{communitynotes}.
While this proposal is inspiring, \cite{chuai2023roll} point out 
that the note validation process was too slow to prevent an effective reduction of misinformation spread.
More importantly, 
the \emph{Community Notes} have been argued to be infiltrated by disinformation groups,
whose coordinated attacks are endangering the value of the system~\cite{communitynotes_wired}.
Unfortunately, building secure community-driven systems 
that appropriately moderate and prioritize information 
is arguably under-researched.
As a result, 
today's algorithms are mostly designed, managed and governed 
in a relatively unilateral and opaque manner.
As exposed by the Facebook Files~\cite{facebook_files},
they are benefiting from an alarming lack of accountability.
Our paper presents a contribution to the algorithmic governance toolbox,
and to the understanding of its challenges.
More precisely, we provide an end-to-end modular pipeline, 
which we instantiate with state-of-the-art algorithms,
to allow any community of non-technical users to securely and collaboratively 
score any number of entities.
More specifically, we make the following contributions.
\paragraph{Contributions.}
Our main contribution is to introduce 
a modular end-to-end collaborative scoring pipeline called \solidago{}\footnote{
  \solidago{} stands for ``Solid Algorithmic Governance''.
  It is also the name of a flower, also known as ``goldenrods''.
}.
Its six modules are 
(1) \emph{trust propagation}, 
(2) \emph{voting rights assignment},
(3) \emph{preference learning},
(4) \emph{model scaling},
(5) \emph{model aggregation}
and (6) \emph{post-process}.
\solidago{} also proposes default implementations for each module,
based on the state of the art.
\begin{figure*}[ht]
    \centering
    \includegraphics[width=\textwidth]{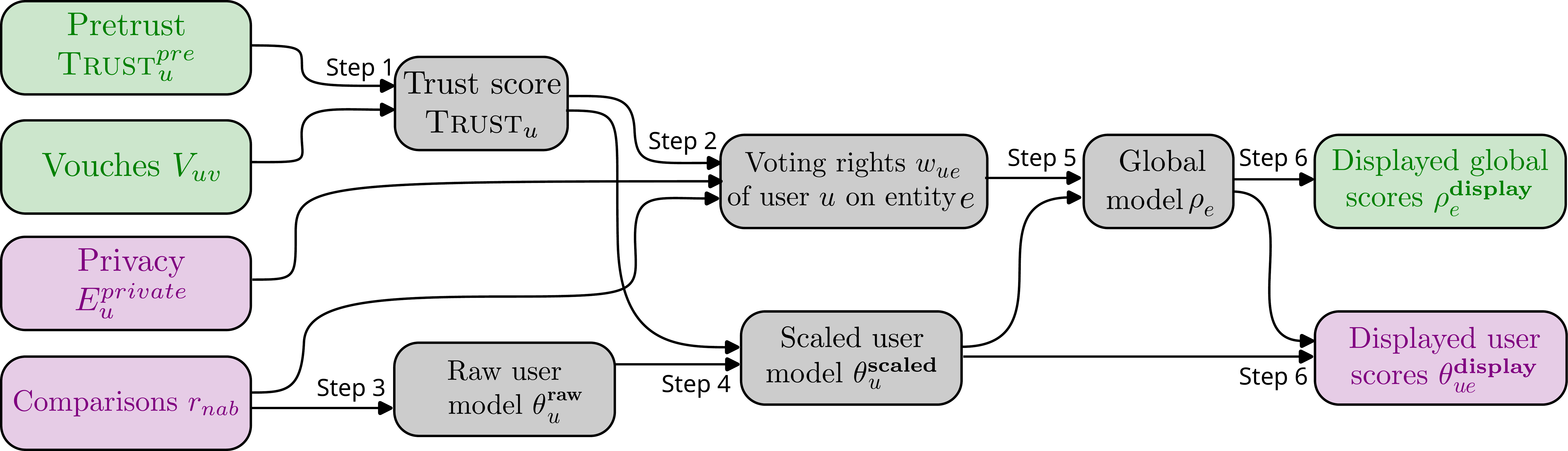}
    \caption{This figure describes the \solidago{} pipeline (slightly simplified). 
    Green boxes correspond to public data, while black boxes are kept hidden.
    The purple boxes contain both public and private data.
    The pipeline is composed of 6 steps, namely 
(1) trust propagation, 
(2) voting rights assignment,
(3) preference learning,
(4) model scaling,
(5) model aggregation
and (6) post-process.
    }
    \label{fig:tournesol_variables}
\end{figure*}
Classical solutions did not always meet the needs of our pipeline.
Our secondary contribution is to fill missing gaps,
by (i) introducing \byztrust{} for secure trust propagation over a directed social network,
(ii) presenting new Lipschitz-resilient primitives~\cite{AllouahGHV22} like \qrqtl{},
(iii) proving their security and consistency guarantees,
(iv) adding a presumption of engagement bias in favor of quality entities,
(v) proposing a mechanism to favor more consensual entities.
Finally, we present an evaluation of the pipeline, both on synthetic and on Tournesol's data.
Better yet, \solidago{} has been fully deployed as the engine of the platform \url{Tournesol.app}~\cite{tournesol} since September 2023,
to collaboratively score the recommendability of YouTube videos.

\paragraph{Related work.}
\label{sec:rel_work}
Social media have raised serious concerns, 
in terms of disinformation~\cite{woolley2020reality,piazza2022fake}, radicalization~\cite{RibeiroO0AM20,bastug2020exploring,facebook_radicalization} and mental health~\cite{gao2020mental,instagram_depression,saha2022social}.
\cite{lin2019existential} even talks of an \emph{existential threat}. 
Conversely, ``good'' information prioritization could advance important causes~\cite{hoang2020science,faucon2021recommendation},
e.g. media literacy, public health and world peace.
But who gets to decide what ``good" is?
WeBuildAI~\cite{LeeKKKYCSNLPP19} laid an inspiring foundation.
This platform essentially allows the stakeholders of a food donation system 
to collaboratively select the recipients of donations.
Similar algorithmic governance systems were proposed by~\cite{NoothigattuGADR18} for autonomous car trolley dilemmas, 
and by~\cite{FreedmanBSDC20} for kidney donation.
Essentially, these systems rely on a collaborative scoring of different options,
which we call \emph{entities}.
The key contribution of our paper is to build a pipeline that generalizes all these systems.
Like \cite{NoothigattuGADR18, LeeKKKYCSNLPP19, FreedmanBSDC20}, 
we consider preference elicitation through entity comparisons~\cite{maystre2018efficient}. 
While comparisons seem more cognitively demanding, 
implying lower user retention,
they may result in more thoughtful judgments,
which have been shown to be less exclusive of outgroup members~\cite{agan2023automating}.
Turning comparisons into scores is then a classical problem~\cite{thurstone1927method}, 
with solutions deployed at scale, e.g. to score chess players~\cite{elo1978rating}. 
\solidago{}'s default implementation
leverages a recent generalization~\cite{GBT2023} of the famous Bradley-Terry model~\cite{BradleyTerry52}.
Contrary to \cite{NoothigattuGADR18, LeeKKKYCSNLPP19, FreedmanBSDC20}, 
\solidago{} does not restrict itself to a predefined list of participants.
This exposes us to fake accounts, also known as \emph{Sybil attacks}~\cite{Douceur02}.
To limit this risk,
we leverage a set of pretrusted users and peer-to-peer vouches.
Trust is then spread through the vouching network~\cite{cheng2005,DanezisM09,TranMLS09,AndersenBCFFKMT08,PoupkoSST21}.
EigenTrust~\cite{KamvarSG03} was proposed to do so, 
and was applied to file sharing~\cite{AbramsMP05,LuWL16}, grid computing~\cite{LiHLW05} and message routing~\cite{SubbarajS14},
and has spurred many variants~\cite{FanLLS12,KalalaFK17,RguibiM19,AfanadorOBA20}.
However, none provides a guarantee on the maximal impact of a single node.
Even in permissioned systems, security issues still arise, 
as any participant may be malicious, corrupted or hacked.
To secure collaborative governance, 
we draw inspiration from the \emph{one person, one unit force} fairness principle~\cite{ElMhamdiFGH21,FarhadkhaniGHV22}.
More specifically, we adapt \textsc{Mehestan},
and provide \emph{Lipchitz-resilience} guarantees~\cite{AllouahGHV22}.
Additionally, our algorithms account for (genuine) users' noisy inputs, 
and output human-friendly scores.
\paragraph{Paper structure.}
In Section~\ref{sec:data}, we present the \solidago{} inputs and problem formulation.
Then, each \solidago{} module is explained in a separate section,
along with our proposed default implementation.
Section~\ref{sec:evaluation} evaluates our pipeline,
and Section~\ref{sec:conclusion} concludes.
Our pipeline is represented in Figure~\ref{fig:tournesol_variables}.
Details are provided in the Appendix.
The code is in the Supplementary Material.
\section{Problem formulation and data}
\label{sec:data}
The problem we consider
consists of defining a modular, well-behaved and secure pipeline 
to transform the inputs from $\Contributor$ different users 
into a human-interpretable score 
$\rho_\alternative^{\bf display} \in \setR$ 
of every entity $\alternative$
in a set $[\Alternative] \triangleq \set{1, \ldots, \Alternative}$ of entities to assess.
Below, we detail \solidago{}'s inputs.
\paragraph{Pretrust.}
To provide Sybil resilience to our system, 
we leverage \emph{pretrust} information about users.
This may be a list of known users.
In the case of Tournesol, pretrust is obtained through the validation of an email,
whose domain is consider Sybil-resilient~\cite{tournesol}.
Typically, \texttt{@who.int} would be considered Sybil-resilient,
but \texttt{@proton.me} would not,
as anyone can easily create multiple such emails.
Formally, we denote $\Contributor^{pre}_\checkmark$ the set of pretrusted users.
\paragraph{Vouch data.}
Since validating pretrusted users only is very exclusive,
\solidago{} also allows users to vouch for one another.
The dataset thereby constructed entails entries of the form $(\contributor, \contributorbis)$, 
which mean that voucher $\contributor$ vouched for vouchee $\contributorbis$.
We denote $\Vouch$ the set of such vouches.
\begin{figure}[h]
    \centering
    \includegraphics[width=\linewidth]{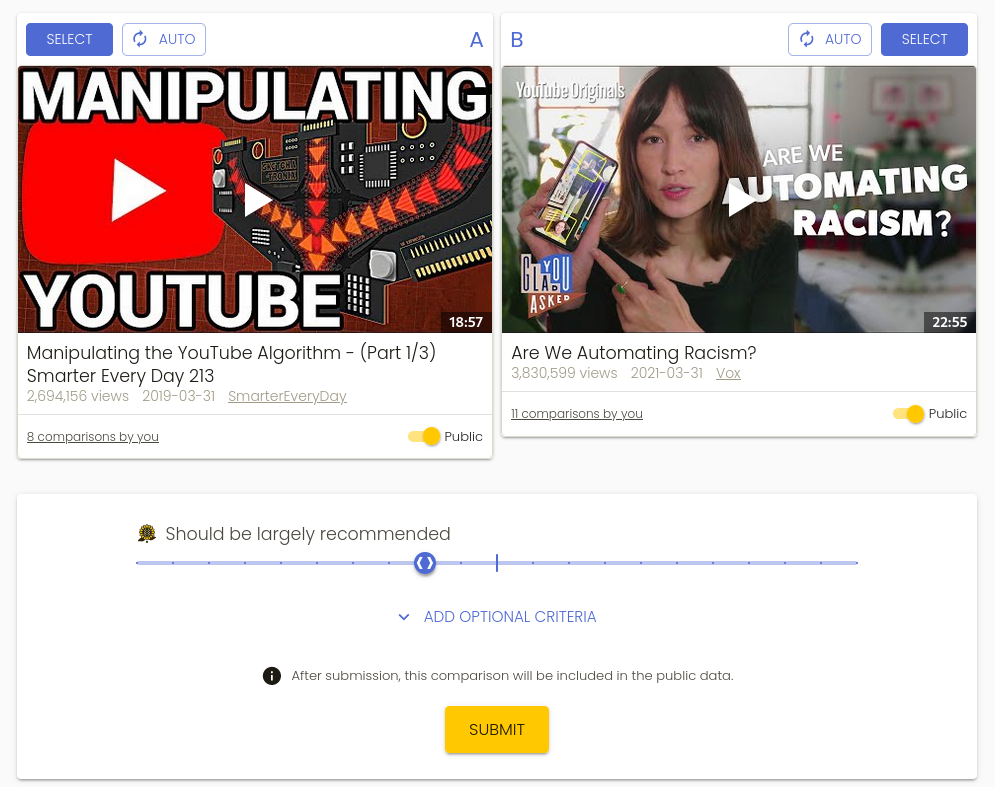}
    \caption{Tournesol's users are asked to select any two videos and to report which they would recommend more widely.}
    \label{fig:comparison_interface}
\end{figure}
\paragraph{Comparisons.}
Each user $u$'s dataset $\mathcal D_\contributor$ is a list of tuples $(\alternative, \alternativebis, \comparison)$,
each saying that contributor $\contributor \in \Contributor$ compared entity $\alternative \in \Alternative$ to entity $\alternativebis \in \Alternative$, 
and gave the judgment $\comparison \in [-\maxComparison, \maxComparison]$.
On Tournesol, the maximal comparison is $\maxComparison \triangleq 10$ (see Figure~\ref{fig:comparison_interface}).
A comparison $\comparison < 0$ means that user $\contributor$ prefers $\alternative$ to $\alternativebis$.
By considering that comparisons are anti-symmetric 
(i.e. reporting $\comparison_{\alternative \alternativebis}$ is equivalent to reporting $- \comparison_{\alternativebis \alternative}$),
we can then represent user $\contributor$'s comparisons
by an anti-symmetric matrix $\comparison_{\alternative \alternativebis} \in \left( [-\maxComparison, \maxComparison] \right) \cup \set{\perp}$.
The special value $\perp$ means that the entry is lacking, and is thus undefined.
\paragraph{Privacy}
Since privacy may be sometimes desirable, 
especially when retaliation risks are significant, 
e.g. when upvoting a content that criticizes a government, an employer or a colleague,
\solidago{} allows each user to assess an entity $\alternative$ either publicly or privately.
However, private comparisons may also be regarded as an attack vector, as such data cannot be audited by third parties.
In the sequel, $\Alternative_\contributor^{private}$ denotes the sets of entities that user $\contributor$  chose to assess privately.
We stress however that \solidago{} does \emph{not} currently provide differential privacy for displayed scores.
\paragraph{Problem formulation.}
Our pipeline constructs a map
$\solidago{} : (\Contributor_\checkmark^{pre}, \Vouch, \mathcal D_{1:\Contributor}, \Alternative_{1:\Contributor}^{private}) \mapsto \rho^{\bf display}$.
We do so by combining six modular steps,
each of which is detailed in the sequel.

\section{Trust propagation}
Step 1 of the \solidago{} pipeline combines pretrusts and vouches to assign trust scores to all users.
Our default implementation combines the introduction of a sink vouch 
and of a novel trust propagation algorithm called \byztrust{}.
\subsection{Sink vouch}
\label{sec:vouch_network}
If vouchers' vouches were unconstrained,
then a malicious user could gain an unbounded influence 
by vouching for arbitrarily many fake accounts.
To secure vouching,
it might be tempting to force vouchers' vouches to be split among their vouchees.
However, this disincentivizes vouching, 
as each new vouch reduces the value of their prior vouches.
To incentivize vouching without sacrificing security, we introduce a \emph{sink vouch} for all users.
Essentially, each user is considered to implicitly vouch for a sink (a non-user),
with an intensity equal to $\sinkvouch$ regular vouchees.
Tournesol set $\sinkvouch \triangleq 5$.
As a result,
when a voucher with much fewer than $\sinkvouch$ vouchees vouches for more vouchees, 
the total vouches given by the voucher grows almost linearly,
thereby not penalizing previous vouchees.
Vouching is thereby not (too) disincentivized.
To integrate sink vouches formally, 
we simply split voucher $\contributor$'s vouchees.
For each $(\contributor, \contributorbis) \in \Vouch$, this amounts to defining
  $\vouch_{\contributor \contributorbis} \triangleq
  1 / \left(
    \sinkvouch + 
    \card{\set{ k \st (n, k) \in \Vouch }}
  \right)$.
Clearly, the sum of weighted vouches given by any voucher to actual users is then at most one.
Thus $\vouch$ is row-substochastic.
\subsection{\byztrust{}}
\label{sec:eigentrust}
Trust propagation through the weighted vouch network, with weights $\vouch$,
is then performed by \byztrust{}.
This novel algorithm is inspired from \eigentrust{}~\cite{KamvarSG03},
whose principle is akin to the celebrated \pagerank{}~\cite{page1999}.
These algorithms have an appealing probabilistic intuition, 
which we present before providing the formal definition.
\paragraph{Probabilistic intuition.}
Consider a random walker on the vouch network, 
initially dropped at a uniformly randomly chosen pretrusted user.
Then at each iteration, the random walker first tosses a biased coin.
With probability $1-\vouchdecay$ (which will be useful for Sybil resilience), 
the random walker resets its walk.
Otherwise, the random walker randomly selects an outgoing vouch, 
with a probability equal to the vouch's weight.
If the arc goes to the sink, then the random walk is also reset.
The frequency of the random walker visits to any given node eventually converges to an ergodic limit,
which defines its trust score.
\byztrust{} essentially robustifies the random walk, 
by frequently preventing the walker from visiting too frequently visited users,
thereby bounding their maximal influence.
\paragraph{Defining \byztrust{}.}
\label{sec:eigentrust_definition}
Recall that we are given a set $\Contributor_\checkmark^{pre}$ of pretrusted users.
We define the pretrust vector by $\pretrust{\contributor} \triangleq 1$ if $\contributor \in \Contributor_\checkmark^{pre}$, and $0$ otherwise.
We then define $\interimtrust{0}{} \triangleq \pretrust{} \in \setR^\Contributor$.
For all iterations $\iteration$, using vector notations, we consider
  $\interimtrust{\iteration + 1/2}{} \triangleq
  \pretrust{}
  + \vouchdecay \vouch^T \interimtrust{\iteration}{}$,
where $\vouchdecay \in (0,1)$ is a vouch decay hyperparameter.
Tournesol set $\vouchdecay \triangleq 0.8$.
Trusts are then clipped to avoid power concentration.
  $\interimtrust{\iteration + 1}{\contributor} \triangleq
  \min \set{ 
  \interimtrust{\iteration + 1/2}{\contributor},
  1 }$.
\begin{proposition}
\label{prop:trust}
    There is a unique vector $\trust{}$ such that
    $\norm{\interimtrust{\iteration}{} - \trust{}}{1} \leq \Contributor \vouchdecay^\iteration$ for all $\iteration$.
    Moreover, this vector verifies
    $\trust{} = \min \set{
    \pretrust{} +
    \vouchdecay \vouch^T \trust{},
    1
    }$.
\end{proposition}
\begin{proof}[Proof sketch]
    Since $\vouch$ is row-substochastic and $x \mapsto \min \set{x, 1}$ is 1-Lipschitz,
    $\interimtrust{\iteration}{}$ is $\vouchdecay$-contractive.
    Hence the convergence.
    Taking the recursion equation of $\interimtrust{\iteration}{}$ to the limit concludes the proof.
\end{proof}
As a corollary, an error of at most $\byztrusterror > 0$ can be guaranteed
by merely $\nbIterations(\varepsilon) \triangleq \left\lceil \frac{\ln(\Contributor / \byztrusterror)}{\ln(1/\vouchdecay)} \right\rceil$ iterations.
Tournesol uses $\byztrusterror \triangleq \eigentrusterrorvalue$.
\paragraph{The security guarantee of \byztrust{}.}
\label{sec:eigentrust_computation}
Our novel algorithm \byztrust{} has the following Sybil resilience guarantee
that bounds the maximal influence of a user by its received trust.
\begin{theorem}
\label{th:trust_security}
Denoting $\trust{}^{-\contributor}$ the trusts computed after removing user $\contributor$'s vouches 
\begin{equation}
    \norm{\trust{} - \trust{}^{-\contributor}}{1}\leq \frac{\vouchdecay}{1-\vouchdecay} \trust{\contributor}^{-\contributor}.
\end{equation}
\end{theorem}
\begin{proof}[Proof sketch]
    The proof considers a similar sequence as $\interimtrust{\iteration}{}$,
    but whose initial point is $\trust{}^{-\contributor}$.
    We conclude by bounding the difference between the first iteration and the initial point,
    and by $\vouchdecay$-contractiveness.
    The full proof is given in Appendix~\ref{app:byztrust}.
\end{proof}
\paragraph{Other properties of \byztrust{}.}
We conclude this section with two additional desirable guarantees,
whose proofs are provided in Appendix~\ref{app:lipschitrust_propositions}.
\begin{proposition}
\label{prop:byztrust_monotone}
  Vouching cannot decrease trust, i.e., for all users $\contributor$, $\trust{} \geq \trust{}^{-\contributor}$.
\end{proposition}
\begin{proposition}
\label{prop:positive_trust}
    A user's trust score is positive, if and only, if there is a vouch path from a pretrusted user to them.
\end{proposition}
\section{Voting rights}
\label{sec:voting-rights}
Step 2 of \solidago{} leverages trust scores and user activity to determine voting rights.
As demanded by Tournesol users,
our solution values untrusted users,
without sacrificing security.
\paragraph{Overview.}
The computation of users' voting rights on an entity $\alternative$ proceeds in four steps:
\begin{itemize}
  \item[1.] Compute the cumulative trusted scores of the users that compared $\alternative$, weighed by $\privacypenalty$.
  \item[2.] Derive the maximal amount of tolerated \emph{overtrust} $\maxovertrust_{\alternative}$.
  \item[3.] Determine a minimal voting right $\minvotingright{\alternative}$ that yields the maximal tolerated overtrust.
  \item[4.] Set voting rights as the maximum of the trust score and $\minvotingright{\alternative}$, weighed by $\privacypenalty$.
\end{itemize}
On Tournesol, the \emph{privacy penalty} is set as $\privacypenalty \triangleq \privacypenaltyvalue$.
\paragraph{Cumulative trusted score}
Recall that we are given the sets $\Alternative_\contributor^{private}$ of entities
that user $\contributor$ rates privately.
From the datasets $\mathcal D_n$, we also know whether the user has rated an entity.
Define $\votingright^{penalty}_{\contributor \alternative} \triangleq 1$ if user $\contributor$ rated entity $\alternative$ publicly,
$\votingright^{penalty}_{\contributor \alternative} \triangleq \privacypenalty$ if user $\contributor$ rated entity $\alternative$ privately,
and $\votingright^{penalty}_{\contributor \alternative} \triangleq 0$ if they did not rate $\alternative$.
The cumulative trusted score is given by
\begin{equation}
  \trust{\alternative} \triangleq
  \sum_{\contributor \in [\Contributor]} \votingright^{penalty}_{\contributor \alternative} \trust{\contributor}
\end{equation}
\paragraph{Maximal tolerated overtrust}
The maximal tolerated overtrust must be asymptotically only a fraction $\FractionOfNonVerified$ of $\trust{\alternative}$,
and will be at least $\defaultnonverifiedvotingright$.
This principle is formalized by defining
\begin{equation}
  \maxovertrust_{\alternative} \triangleq
  \defaultnonverifiedvotingright
  + \FractionOfNonVerified \cdot \trust{\alternative}.
\end{equation}
Tournesol set $\defaultnonverifiedvotingright \triangleq \defaultnonverifiedvotingrightvalue$ and 
$\FractionOfNonVerified \triangleq \FractionOfNonVerifiedValue$.
\paragraph{Minimal voting right}
We now guarantee that the minimal voting right given to insufficiently trusted users yields at most an overtrust of $\maxovertrust_{\alternative}$.
To do so, define the \emph{overtrust} when setting a minimal public voting right $\minvotingright{\alternative}$ by
\begin{align}
    \overtrust_{\alternative} (\votingright_{min}) 
    \triangleq 
    \sum_{\contributor \in [\Contributor]}
    \votingright^{penalty}_{\contributor \alternative} \max \set{\votingright_{min} - \trust{\contributor}, 0} 
\end{align}
Observe that we have $\overtrust_{\alternative} (0) = 0$ 
and that  $\overtrust_{\alternative}$ is a strictly increasing function 
of $\minvotingright{\alternative}$ as soon as $\overtrust_{\alternative} (\minvotingright{\alternative}) > 0$.
Thus, given that the maximal overtrust is positive, 
either $\overtrust_{\alternative} (1) < \maxovertrust_{\alternative}$ (in which case we set $\minvotingright{\alternative} \triangleq 1$),
or there is a unique $\minvotingright{\alternative} \in (0,1]$ for which we have $\overtrust_{\alternative} (\minvotingright{\alternative}) = \maxovertrust_{\alternative}$.
We propose to compute it with dichotomic search.
\paragraph{Voting rights}
A user $\contributor$'s voting right on $\alternative$ is then defined
as the maximum between their trust and the public minimal voting right for this entity,
potentially multiplied by the privacy penalty, i.e.
\begin{equation}
    \votingright_{\contributor \alternative} \triangleq 
    \votingright_{\contributor \alternative}^{penalty} \max \left\lbrace 
    \trust{\contributor}, \minvotingright{\alternative}
    \right\rbrace.
\end{equation}

\section{User model inference}
\label{sec:raw_scores}
Step 3 of \solidago{}, which is independent from the two first steps,
consists of learning a preference model for each user $\contributor \in [\Contributor]$,
based on their entity-evaluation dataset $\mathcal D_\contributor$.
Our default implementation transforms quantitative comparisons into scores for compared entities.
\paragraph{Scoring through Generalized Bradley-Terry.}
To turn a user's comparisons into scores,
we use the recent generalization~\cite{GBT2023} of~\cite{BradleyTerry52},
which is instantiated with a uniform root law.
This yields a negative log-likelihood equal to
\begin{equation}
\label{eq:contributorloss}
    \mathcal L \left( \contributorscore{}_\contributor | \mathcal D_\contributor \right)
    \triangleq 
    \sum_{\alternative, \alternativebis, \comparison \in \mathcal D_\contributor} \ln \frac{\sinh{} (\contributorscore_{\contributor \alternative \alternativebis})}{\contributorscore_{\contributor \alternative \alternativebis}} + \frac{\comparison \contributorscore_{\contributor \alternative \alternativebis}}{\maxComparison}.
\end{equation}
Including a Gaussian prior yields the following loss, which we minimize:
\begin{equation}
    \optcontributorscore_\contributor(\comparison)
    \triangleq \argmin{\contributorscore_\contributor \in \setR^{\Alternative_\contributor}}{
        \frac{\contributorpriorweight}{2} \norm{\contributorscore_\contributor}{2}^2
        + \mathcal L \left( \contributorscore{}_\contributor | \mathcal D_\contributor \right)
    }.
\end{equation}
Tournesol uses $\contributorpriorweight \triangleq 0.02$.
Since \cite{GBT2023} guarantees that the loss is strongly convex, 
the minimum $\optcontributorscore(\comparison)$ exists, is unique and can be efficiently computed,
which we call \emph{raw scores}.
This model also has the following desirable guarantee.
\begin{proposition}[Theorem 2 in \cite{GBT2023}]
    \label{prop:monotonicity}
If $\comparison_{\alternative\alternativebis} \leq \comparison_{\alternative\alternativebis}' $ 
for all $\alternativebis$ (i.e., the comparisons are more favorable to $\alternative$ in $\comparison'$ than in $\comparison$)
and $\normalizedcomparison_{\alternativeter\alternativequater} = \normalizedcomparison_{\alternativeter\alternativequater}' $ for all $\alternativeter, \alternativequater \neq \alternative$, 
then $\optcontributorscore_\alternative(\comparison) \leq \optcontributorscore_\alternative(\comparison')$.
\end{proposition}
\paragraph{Estimating uncertainty.}
\label{sec:raw_score_uncertainty}
Learning the uncertainty of the scores is important for score aggregation; 
intuitively, a user should have less influence on the global scores if we have a high uncertainty on their scores.
We perform an asymmetric uncertainty estimation, 
based on increase of the negative log-likelihood term of the loss by one unit 
(which corresponds to a likelihood $1/e$ times smaller),
i.e. we solve for $\optcontributorscoreuncertainty_{\contributor \alternative, left} \geq 0$
and $\optcontributorscoreuncertainty_{\contributor \alternative, right} \geq 0$ such that
\begin{align}
    \mathcal L \left( 
        \contributorscore{}_\contributor - \optcontributorscoreuncertainty_{\contributor \alternative, left} {\bf e}_{\alternative} | \mathcal D_\contributor 
    \right)
    = 
    \mathcal L \left( 
        \contributorscore{}_\contributor + \optcontributorscoreuncertainty_{\contributor \alternative, left} {\bf e}_{\alternative} | \mathcal D_\contributor 
    \right)
    = 
    \mathcal L \left( 
        \contributorscore{}_\contributor | \mathcal D_\contributor 
    \right) + 1,
\end{align}
where ${\bf e}_{\alternative}$ is the $\alternative$-th vector of the canonical basis.
Note that such uncertainties may not exist, 
in which case they are set to $+\infty$.  
\section{Model scaling}
\label{sec:scaling}
Step 4 of \solidago{} performs user model scaling.
The default implementation has three steps, 
namely finding a common scale, 
accounting for scoring selection bias
and standardizing scores.
\subsection{Collaborative scaling}
\label{sec:mehestan}
As explained by~\cite{AllouahGHV22}, 
different users express themselves on different implicit scales,
which implies scale mismatch.
This issue may especially be raised if a user assesses top entities,
while another assesses bad entities only.
To put all users' scores on a common scale in a secure manner,
we adapted \mehestan{}~\cite{AllouahGHV22}.
\paragraph{Reviewing \mehestan{}.}
The original \mehestan{} has three steps:
\begin{itemize}
  \item[(i)] It individually normalizes all users' scores.
  \item[(ii)] Each user's scores are scaled, by comparing them to other users' scores for entities that both compared, using a Lipschitz-resilient primitives.
  \item[(iii)] Third, each users' scores are appropriately translated, by comparing them to other users' scores, again using a Lipschitz-resilient primitives.
\end{itemize}
\paragraph{How we adapted \mehestan{}.}
We diverge from~\cite{AllouahGHV22} in four notable ways:
\begin{itemize}
  \item[(1)] We skip step (i), as comparisons provides a reasonable multiplicative normalization.
  \item[(2)] We rely on a subset of \emph{scaling-calibration} users (detailed below).
  Other users' scales are then made to match scaling-calibration users' scale.
  \item[(3)] We robustify the estimates of the relative scaling between two users, 
  by using the primitive $\qrmed{}$ instead of an average, 
  and generalizing other primitives to account for \emph{asymmetric uncertainties} (see Appendix~\ref{app:primitives}).
  \item[(4)] We systematically track uncertainties on estimates, 
  which are used when aggregating them.
  Intuitively, high uncertainty on a scaled user score for an entity 
  implies that it will only have an effect
  if the aggregated value is farther away from the estimate than the computed uncertainty.
  We thereby contribute to the growing toolbox of Lipschitz-resilient primitives.
\end{itemize}
Our adaptation of \mehestan{} is detailed in Appendix~\ref{app:mehestan}.
\paragraph{Scaling-calibration users.}
\label{sec:voting_rights_scaling}
Adaptation (2) is motivated by the ``liveness'' guarantees of \mehestan{},
which requires users to be sufficiently active.
This is far from being the case in general, and in particular in the case of Tournesol.
Ignoring inactive users may be reminiscent of~\cite{garin2022weighting}, 
who propose to ignore insufficiently active participants to federated learning.
Additionally, (2) allows to accelerate \mehestan{}, 
whose vanilla form is too computationally demanding for frequent score updates,
especially given Tournesol's limited server with a 4-core CPU and no GPU.
To make collaborative preference scaling more accurate and faster,
we propose to select the $N^{scaling}$ users who have compared the most entities, 
among those whose trust score is at least $\trust{min}^{scaling}$. 
The values $N^{scaling}=100$ and $\trust{min}^{scaling} = 0.1$ are currently used on the platform.
\subsection{Presumption of quality of rated entities}
In the context of Tournesol, 
it has been observed that most of the entities that are rated on the platform are of high quality.
This raises an issue.
Namely, since zero is the default score in absence of comparisons,
an unrated entity would be regarded as equally good as an average rated entity.
To account for the selection bias among rated entities,
we encode a \emph{presumption of quality} for rated entity,
by shifting the scale of the scores of rated entities,
so that the $\zeroshiftquantile$-quantile of rated entities has a zero score.
Tournesol set $\zeroshiftquantile \triangleq 15\%$, 
i.e. they consider that $15\%$ of rated entities are less recommendable than a random YouTube video.
More precisely, to do this robustly, 
we define a new Lipschitz-resilient primitive, 
which we call the $q$-\emph{quadratically regularized quantile} (see Section~\ref{sec:model_aggregation}).
We then subtract to all user scores
the $\zeroshiftquantile-\qrqtl{}$ of the users' scores, given their uncertainties.
\subsection{Score standardization to control Lipschitz resilience}
Finally,
we standardize the scaled scores using a Lipchitz-resilient estimator of their standard deviation.
Note that Tournesol's \emph{collaboratively scaled scores} are very heavy-tailed,
with many values that are over 30 standard deviations away from the mean. 
This motivated us to use an estimator that intuitively fits $q^{std}_{dev} \triangleq 90\%$ of all the scores.
Essentially, this amounts to acknowledging that at most $10\%$ of the entities will have scores
that are out of scale .
More precisely, we consider the $q^{std}_{dev}-\qrqtl{}$ of the absolute deviations of the scores to the $\qrmed{}$ of the scores\footnote{
  Each user score is given a voting right inversely proportional to the number of entities the user compared.
}.
The scaled scores are standardized by dividing them by the resulting estimate $\sigma^{\bf scaled}$ of the standard deviation.
The resulting scores are the \emph{scaled user models} $\theta_{\contributor \alternative}^{\bf scaled}$,
with uncertainties $\Delta \theta_{\contributor \alternative}^{\bf scaled}$,
\section{Model aggregation}
\label{sec:model_aggregation}
Step 5 of \solidago{} consists of aggregated users' scaled models.
For each entity $\alternative$, 
we aggregate the users' \emph{scaled scores} $\theta^{\bf scaled}_{\contributor \alternative}$
using a novel primitive we call the \emph{quadratically regularized quantile}.
Namely, for $\alpha \in (0, 1)$, voting rights $\votingright_\contributor$, scalars $x_\contributor$ and uncertainties $\Delta_\contributor$ for $\contributor \in [\Contributor]$,
we define
\begin{equation}
\label{eq:qrqtl}
  \qrqtl{}_{\alpha, \lipschitz} (\votingright, x, \uncertainty)
  \triangleq \argmin{m \in \setR}{
    \frac{m^2}{2 \lipschitz} 
    + \sum_{\contributor : x_\contributor \neq \perp} 
        \votingright_\contributor 
        \huber_\alpha(m|x_\contributor, \uncertainty_\contributor)
  },
\end{equation}
where $\huber$ is an adaptation of the classical Huber loss to asymmetric uncertainties:
\begin{equation}
    \huber_\alpha(m|x, \uncertainty) \triangleq \left\lbrace
    \begin{array}{ll}
        \min \set{ 1, \frac{1-\alpha}{\alpha}} \left( 
            \sqrt{\uncertainty_{left}^2 + (x-m)^2} - \uncertainty_{left} 
        \right) & \text{for } m \leq x,  \\
        \min \set{ 1, \frac{\alpha}{1-\alpha}} \left( 
            \sqrt{\uncertainty_{right}^2 + (m - x)^2} - \uncertainty_{right}
        \right) & \text{for } m \geq x.
    \end{array}
    \right.
\end{equation}
Note that, for $\alpha = 0.5$ and $\Delta = 0$,
we retrieve the quadratically regularized median as defined by~\cite{AllouahGHV22}.
Crucially, for any $\alpha \in (0, 1)$, 
$\qrqtl_{\alpha, \lipschitz}$ remains $\lipschitz$-Lipschitz resilient,
\begin{theorem}[Lipschitz resilience of \qrqtl{}]
    Let $\alpha \in (0, 1)$.
    For any $x_{1:\Contributor} \in \setR^{U}$, $\Delta_{1:\Contributor} \in \setR_{\geq 0}^{2 \cdot U}$,
    the map $\votingright_{1:\Contributor} \mapsto \qrqtl{}_{\alpha, \lipschitz} (\votingright_{1:\Contributor}, x_{1:\Contributor}, \Delta_{1:\Contributor})$
    is $\lipschitz$-Lipschitz continuous as a function $(\setR^\Contributor, \norm{\cdot}{1}) \rightarrow (\setR, \absv{\cdot})$.
\end{theorem}
\begin{proof}[Proof sketch]
    The key to the theorem is the derivative of the term that depends on a user $\contributor$ has an absolute value which is at most $\votingright_\contributor$.
    See Appendix~\ref{app:qr_quantile} for the details.
\end{proof}
Appendix~\ref{app:qr_quantile} proves several other desirable properties of \qrqtl{},
such as the fact that it asymptotically correctly estimates the quantile,
or the fact that the addition of a user with huge uncertainties does not affect the estimate.
We then define each \emph{global score} $\rho_{\alternative}$ for entity $\alternative$ as
\begin{equation}
\rho_\alternative \triangleq \qrqtl_{\alpha, \lipschitz} \left(
  \votingright_{1:\Contributor, \alternative}, 
  \frac{\theta^{\bf scaled}_{1:\Contributor, \alternative}}{\sigma^{\bf scaled}}, 
  \frac{\Delta \theta^{\bf scaled}_{1:\Contributor, \alternative}}{\sigma^{\bf scaled}}
\right).    
\end{equation}
Tournesol set $\alpha \triangleq 0.2$ and $\lipschitz \triangleq 0.1$.
\paragraph{Favoring consensus.}
By setting $\alpha < 0.5$, Tournesol favors entities that are consensually good,
over entities that are excellent according to half of the community,
but are very bad according to another half.
This measure is aligned, in spirit, with some measures of Twitter's \emph{Community Notes}.
Appendix~\ref{app:quantile_change} discusses the impact of $\alpha$.
\section{Model post-process}
\label{sec:displayed_scores}
Finally, Step 6 of \solidago{} responds to the desire for a display of scores on a bounded scale, e.g. from $(-100, 100)$,
which may be regarded as more ``human-readible''.
We do so by squashing the scores 
using the increasing homeomorphism $\varphi: x \mapsto 100 x / \sqrt{1 + x^2}$,
which transforms $\setR$ into $(-100, 100)$.
This yields the \emph{displayed user scores} $\theta^{\bf display}_{\contributor \alternative}$
and the \emph{displayed global scores} $\rho^{\bf display}_{\alternative}$.
These are formally given by
$\theta^{\bf display}_{\contributor \alternative}
    \triangleq \varphi \left( \theta^{\bf scaled}_{\contributor \alternative} / \sigma^{\bf scaled} \right)$
and $\rho^{\bf display}_{\alternative} \triangleq \varphi(\rho_\alternative)$.
\section{Evaluation}
\label{sec:evaluation}
We perform two evaluations of our pipeline.
First we discuss experiments made with a synthetic dataset, 
before discussing the large-scale deployment on Tournesol.
\subsection{Evaluation on synthetic data}
\paragraph{User generation.}
To evaluate \solidago{} on synthetic data,
we propose a user generation model,
where each user may be trustworthy or not, 
pretrusted or not.
Here, we make the simplifying assumption that pretrusted users are trustworthy.
Moreover, a user may be more or less engaged in vouching for others,
and in comparing entities.
They may also be more or less willing to provide numerous comparisons per entity.
We draw a singular-value decomposition (SVD) model of users and entities.
Namely, we draw a preference $d$-dimensional SVD vector for each trustworthy user
from a normal distribution $\mathcal N(\rho^{\bf true}, I_d)$,
with $\rho^{\bf true} \triangleq (3, 0, \ldots, 0)$ 
(and $d=2$ in our experiments).
Note that $\rho^{\bf true} \neq 0$ implies that the trustworthy users' preferences are somewhat aligned.
Untrustworthy users are given the vector representation $- \rho^{\bf true}$,
i.e. they score items oppositely from the trustworthy users' approximately consensual preferences.
Finally, some users may prefer to evaluate good videos rather than any video.
We model this with a random engagement bias for good videos.
\paragraph{Vouch generation.}
We adapt the~\cite{erdds1959random} model for vouching among trustworthy users,
where the probability of vouching is however dependent on the voucher's activeness in vouching.
Similarly, untrustworthy users vouch for one another, according to a similar model.
\paragraph{Entity generation.}
Each entity is assumed to have a feature SVD vector drawn from $\mathcal N(0, I_d)$.
\paragraph{Engagement generation.}
For each user, we define the engagement bias as
the user's implicit score for an entity scaled by their engagement bias, 
plus a random standard normal noise.
The user then engages with the entities with the best engagement score.
Among these entities, the comparisons are made based on an Erdös-Renyi model.
\paragraph{Comparison generation.}
Each engagement leads to a comparison, made by a $21$-nary generalized Bradley Terry model~\cite{GBT2023} (as Tournesol does),
given the true score defined by the scalar product between the user preference's SVD vector
and the entity's SVD vector. 
\paragraph{Reproducibility.}
The plots of Figure~\ref{fig:resilience} depict the standard deviations after running 100 random seeds.
The code is provided in the Supplementary Material.
\begin{figure}[h]
    \centering
    \includegraphics[width=0.49\linewidth]{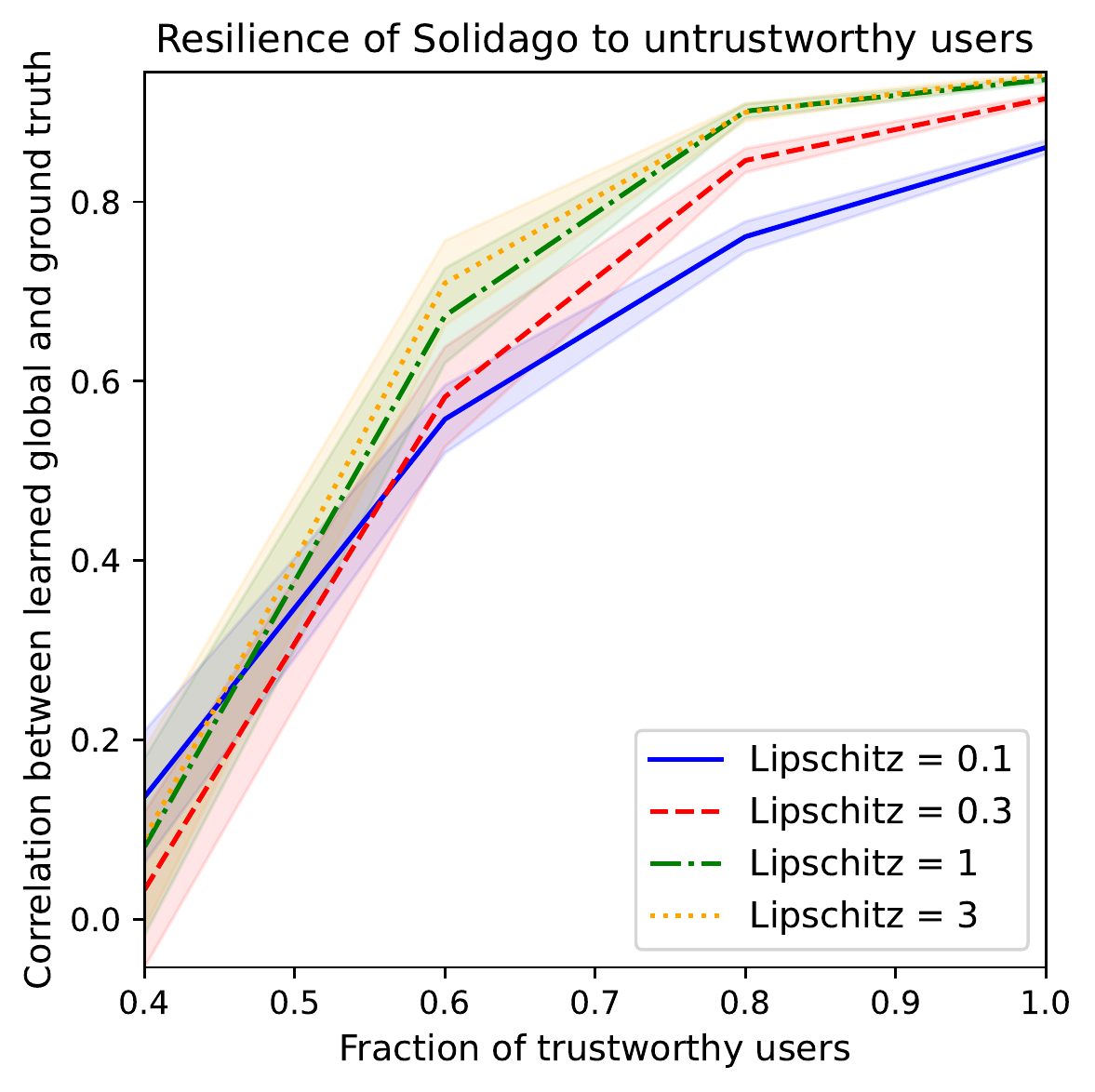}
    \includegraphics[width=0.49\linewidth]{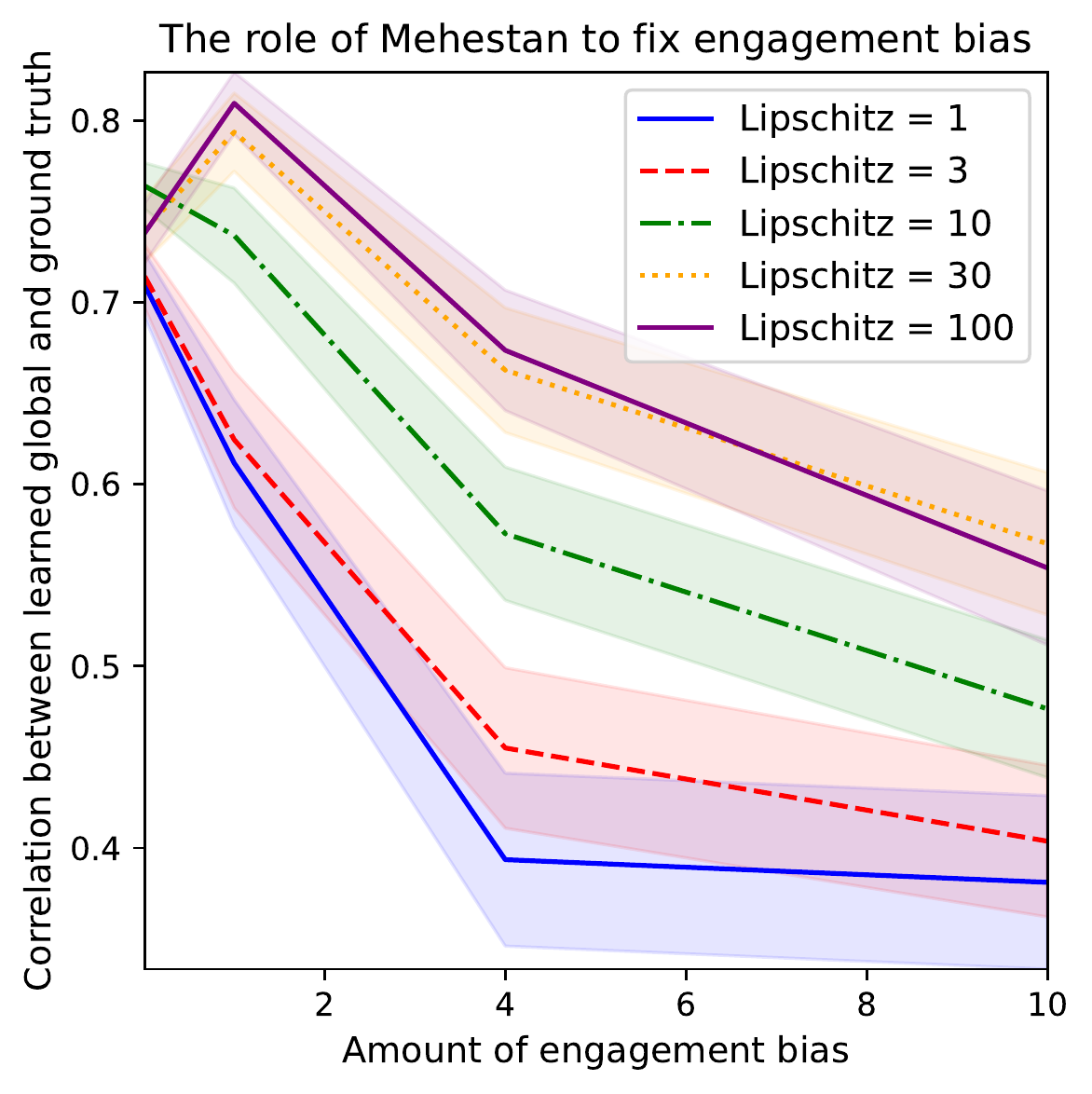}
    \caption{Correlation between the ground truth and the learned global scores, as a function of the fraction of trustworthy users (left),
    and as a function of engagement bias (right).}
    \label{fig:resilience}
\end{figure}
\paragraph{Resilience evaluation.}
Figure~\ref{fig:resilience} (left) shows how the correlation 
between the ground truth scores and the pipeline output
increases with the fraction of trustworthy users.
For a small Lipschitz constant (which corresponds to a more resilience),
this correlation is positive, even given a majority of trustworthy users,
thanks to pretrust, vouching and the pipeline resilience.
When most users are trustworthy, larger Lipschitz constants improve accuracy.
\paragraph{Evaluation of the correction of engagement bias.}
The right plot of Figure~\ref{fig:resilience} shows the harm of engagement bias,
especially Mehestan's Lipschitz resilience is too small
(which amounts to doing barely any scaling).
Overall, the plots highlight the tradeoff between resilience and accuracy.
\begin{figure}[h]
    \centering
    \includegraphics[width=\linewidth]{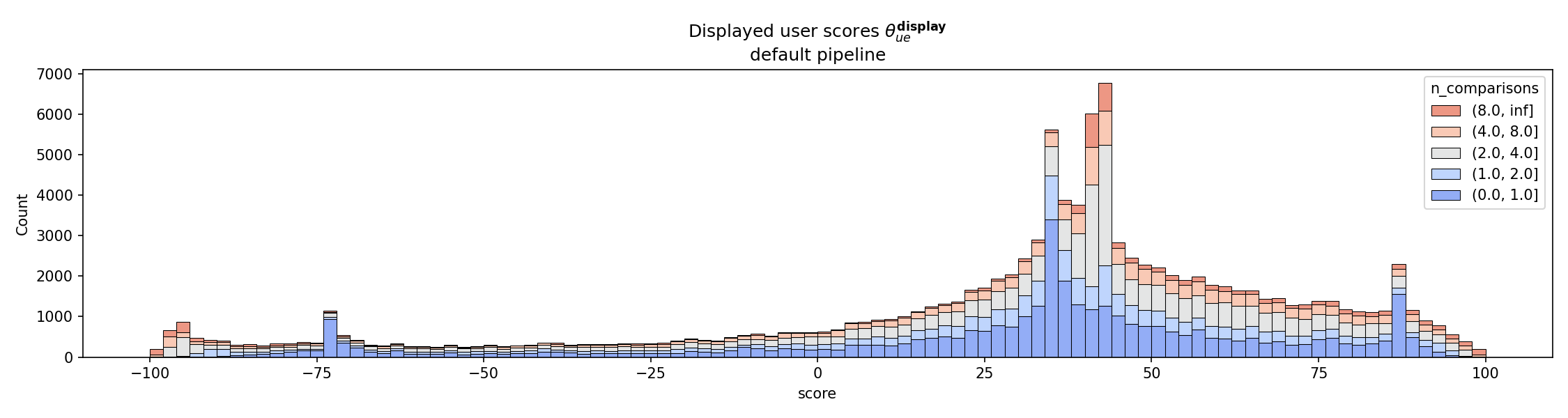} \\
    \includegraphics[width=\linewidth]{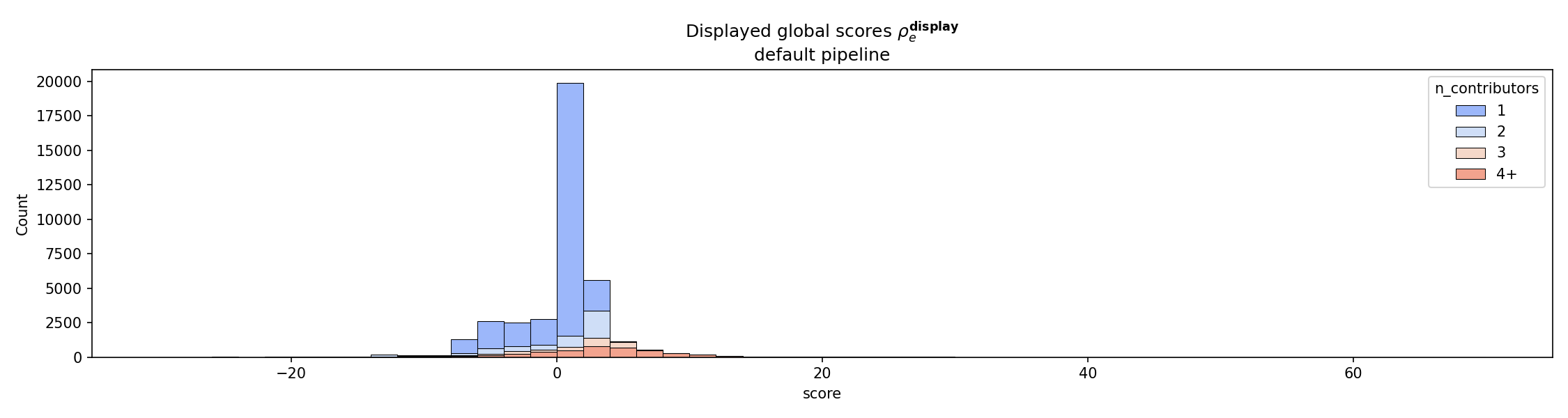} 
    \caption{Distribution of displayed user scores $\theta^{\bf display}_{\contributor \alternative}$ (top graph) and global $\rho^{\bf display}_{\alternative}$ (bottom graph). 
    Hyperparameters were selected so that users' scores have roughly a similar distribution, 
    once they were compared at least 3 times,
    and so that the global scores of highly compared videos are well spread.
    Hyperparameter selection is further discussed in Appendix~\ref{sec:section}.}
    \label{fig:sss_scores}
\end{figure}
\subsection{Large-scale deployment}
\solidago{} has been deployed on Tournesol since September 2023,
where it allows over 10,000 active users
to collaboratively score over 35,000 YouTube videos,
by aggregating their over 190,000 comparisons.
Figure~\ref{fig:sss_scores} displays the distribution of the displayed user and global scores on the platform.
Appendix~\ref{sec:section} details our hyperparameter selection, by plotting the distributions for other hyperparameters.
Note that we heavily exploited \texttt{numba} to accelerate the computations.
While no feedback elicitation has been performed,
the deployment of \texttt{solidago} was well received by several active users,
who reckoned a significant improvement of the algorithmic interpretation of their inputs by the platform.

\section{Conclusion}
\label{sec:conclusion}
This paper presented \solidago{}, 
an end-to-end modular pipeline for secure collaborative scoring.
The pipeline is composed of 6 modules, 
for which we introduced state-of-the-art default implementations.
We made these modules somewhat independent, 
to facilitate external contributions to the pipeline.
In fact, improving each module is an exciting research problem of its own,
involving, e.g., accounting for direct assessments,
learning volitions~\cite{DBLP:conf/hci/LechiakhM23},
enabling liquid democracy~\cite{DBLP:conf/sigecom/0002HJMPR23}, 
and making user models generalize (see Appendix~\ref{sec:future}).
\solidago{} is also accompanied with data generation tools,
to facilitate pipeline evaluation.
Perhaps most remarkably, 
\solidago{} has been deployed at scale on \url{tournesol.app} since September 2023.

\bibliographystyle{plain}
\bibliography{references}

\appendix
\newpage
\onecolumn
\begin{center}
    {\huge Appendix}
\end{center}
\section{Voting rights (algorithms and proofs)}
In this section, we provide a pseudocode of \byztrust{}, and we prove Proposition~\ref{prop:trust} and Theorem~\ref{th:trust_security}.
\subsection{Proof of Proposition~\ref{prop:trust}}
First, let us prove the following useful lemmas.
\begin{lemma}
\label{lemma:min_lipschitz}
$x \mapsto \min \set{x, 1}$ is 1-Lipschitz.
\end{lemma}
\begin{proof}
Let $x,y \in \setR$. Without loss of generality, assume $x \geq y$.
Then we have three cases to analyze:
\begin{itemize}
    \item If $x \geq y \geq 1$, then $\absv{\min \set{x, 1} - \min \set{y, 1}} = \absv{x-y} \leq \absv{x-y}$.
    \item If $x \geq 1 \geq y$, then $\absv{\min \set{x, 1} - \min \set{y, 1}} = \absv{1-y} \leq \absv{x-y}$.
    \item If $1 \geq x \geq y$, then $\absv{\min \set{x, 1} - \min \set{y, 1}} = \absv{1-1} = 0 \leq \absv{x-y}$.
\end{itemize}
In all three cases, 1-Lipschitzness is guaranteed.
\end{proof}
\begin{lemma}
\label{lemma:column_stochastic}
  Suppose $M \in \setR_+^{\Contributor \times \Contributor}$ is column-substochastic,
  i.e. $\sum_{\contributor} M_{\contributor \contributorbis} \leq 1$ for all column $\contributorbis \in \Contributor$ and $M_{\contributor \contributorbis} \geq 0$ for all $\contributor, \contributorbis \in \Contributor$.
  Then for any $x \in \setR^\Contributor$, we have $\norm{Mx}{1} \leq \norm{x}{1}$.
\end{lemma}
\begin{proof}
$\norm{M x}{1}
= \sum_{\contributor} \absv{\sum_{\contributorbis} M_{\contributor \contributorbis} x_\contributorbis}
\leq \sum_{\contributor} \sum_{\contributorbis} M_{\contributor \contributorbis} \absv{x_\contributorbis}
= \sum_{\contributorbis} \left( \sum_{\contributor} M_{\contributor \contributorbis} \right) \absv{x_\contributorbis}
\leq \sum_{\contributorbis} \absv{x_\contributorbis} = \norm{x}{1}$.
\end{proof}
Since $\vouch$ is row-substochastic, we thus have $\norm{\vouch^T x}{1} \leq \norm{x}{1}$.
This allows us to prove Proposition~\ref{prop:trust}.
\begin{lemma}
\label{lemma:trust_contractive}
    The sequence $\interimtrust{\iteration}{}$ is $\vouchdecay$-contractive.
\end{lemma}
\begin{proof}
Denote $\Delta \interimtrust{\iteration}{} \triangleq \interimtrust{\iteration + 1}{} - \interimtrust{\iteration}{}$.
Using Lemma~\ref{lemma:min_lipschitz}, for all $\contributor \in \Contributor$,
we have 
\begin{align}
    \absv{\Delta \interimtrust{\iteration}{\contributor}}
    &= \absv{\interimtrust{\iteration + 1}{\contributor} - \interimtrust{\iteration}{\contributor}} 
    = \absv{\min \set{\interimtrust{\iteration + 1/2}{\contributor}, 1} - \min \set{\interimtrust{\iteration - 1/2}{\contributor}, 1}}
    \leq \absv{\interimtrust{\iteration + 1/2}{\contributor} - \interimtrust{\iteration - 1/2}{\contributor}} \\
    &= \absv{ 
        \left( \pretrust{\contributor} + \vouchdecay \sum_{\contributorbis} \vouch_{\contributorbis \contributor} \interimtrust{\iteration}{\contributorbis} \right) 
        - \left( \pretrust{\contributor} + \vouchdecay \sum_{\contributorbis} \vouch_{\contributorbis \contributor} \interimtrust{\iteration - 1}{\contributorbis} \right) 
    } \\
    \label{eq:bound_delta_interimtrust}
    &= \vouchdecay \absv{ \sum_{\contributorbis \in \Contributor} \vouch_{\contributorbis \contributor} \Delta \interimtrust{\iteration-1}{\contributorbis}}
    \leq \vouchdecay \sum_{\contributorbis \in \Contributor} \vouch_{\contributorbis \contributor} \absv{ \Delta \interimtrust{\iteration-1}{\contributorbis}}
    = \vouchdecay (\vouch^T)_\contributor \absv{\Delta \interimtrust{\iteration - 1}{}},
\end{align}
where $\absv{x}$ of a vector $x$ is the vector whose coordinates are the absolute values of the coordinates of $x$.
It follows that
\begin{align}
    \norm{\Delta \interimtrust{\iteration}{}}{1}
    = \norm{\interimtrust{\iteration + 1}{} - \interimtrust{\iteration}{}}{1}
    = \sum_{\contributor \in \Contributor} \absv{\interimtrust{\iteration + 1}{\contributor} -  \interimtrust{\iteration}{\contributor}} 
    \leq \vouchdecay \norm{ \vouch^T \absv{\Delta \interimtrust{\iteration - 1}{}} }{1}.
\end{align} 
By Lemma~\ref{lemma:column_stochastic}, we thus have $\norm{\interimtrust{\iteration}{}}{1} \leq \vouchdecay \norm{\Delta \interimtrust{\iteration - 1}{}}{1}$.
This amounts to saying that $\interimtrust{\iteration}{}$ is $\vouchdecay$-contractive. 
\end{proof}
\begin{repproposition}{prop:trust}
There is $\trust{}$ such that
$\norm{\interimtrust{\iteration}{} - \trust{}}{1} \leq \card{\Contributor} \vouchdecay^\iteration$.
Moreover, $\trust{} = \min \set{ \pretrust{} + \vouchdecay \vouch^T \trust{}, 1 }$.
\end{repproposition}
\begin{proof}
Lemma~\ref{lemma:trust_contractive} implies its convergence to a limit $\trust{}$.
Taking the recursive defining equation of $\interimtrust{\iteration}{}$ to the limit yields
\begin{equation*}
    \trust{} = \min \set{
    \pretrust{} +
    \vouchdecay \vouch^T \trust{},
    1
    }.
\end{equation*}
Using again Lemma~\ref{lemma:min_lipschitz}, 
for $\iteration \geq 0$, we then have
\begin{equation}
    \norm{\interimtrust{\iteration + 1}{} - \trust{}}{1} 
    \leq \vouchdecay \norm{ 
        \left( \pretrust{} + \vouchdecay \vouch^T \interimtrust{\iteration}{} \right) 
        - \left( \pretrust{} + \vouchdecay \vouch^T \trust{} \right) 
    }{1} 
    \leq \vouchdecay \norm{\vouch^T \absv{\interimtrust{\iteration}{} - \trust{}}}{1}
\end{equation}
Using now Lemma~\ref{lemma:column_stochastic} then yields
\begin{equation}
    \norm{\interimtrust{\iteration + 1}{} - \trust{}}{1} 
    \leq \vouchdecay \norm{ \interimtrust{\iteration}{} - \trust{}}{1}.
\end{equation}
By a straightforward induction, it follows that
\begin{equation}
    \norm{\interimtrust{\iteration}{} - \trust{}}{1} 
    \leq \vouchdecay^\iteration \norm{ \interimtrust{0}{} - \trust{}}{1}
    = \vouchdecay^\iteration \norm{ \pretrust{} - \trust{} }{1}.
\end{equation}
Using triangle inequality yields
\begin{equation}
    \norm{\interimtrust{\iteration}{} - \trust{}}{1} 
    \leq \vouchdecay^\iteration \norm{ \pretrust{} - \trust{}}{1}
    = \vouchdecay^\iteration \sum_{\contributor \in \Contributor} \absv{ \pretrust{} - \trust{\contributor}}
    \leq \vouchdecay^\iteration \card{\Contributor},
\end{equation}
given that each term $\interimtrust{0}{\contributor} - \trust{\contributor}$ is necessarily in $[-1,1]$, 
and thus has an absolute value at most 1.
\end{proof}
\subsection{Proof of Theorem~\ref{th:trust_security}}
\label{app:byztrust}
Intuitively the worst case holds when the Byzantine vouches for a large number of fake accounts without certified email or authentic vouches,
which themselves each vouch for a large number of fake accounts without certified email or authentic vouches, and so on.
Our formal proof instead relies on the study of the sequence $x^\iteration$,
which is defined by initialization $x^0 \triangleq \trust{}^{-\contributor}$ (where contributor $\contributor$ is the one whose influence will be studied)
and recursion $x^{\iteration+1} \triangleq \min \left\lbrace \pretrust{} + \vouchdecay \vouch^T x^{\iteration} \right\rbrace$ (the same vector recursion as for $\interimtrust{}{}$).
This sequence intuitively corresponds to the computation of $\trust{}$, starting from the case where contributor $\contributor$ vouched only for the sink.
\begin{lemma}
\label{lemma:x_convergence}
    $x^\iteration \rightarrow \trust{}$ as $\iteration \rightarrow \infty$.
\end{lemma}
\begin{proof}
Lemma~\ref{lemma:trust_contractive} actually applies to $x^\iteration$ as well, since it only leverages the recursion of $\interimtrust{\iteration}{}$,
which is the same as that of $x^\iteration$.
This not only implies that $x^\iteration$ converges, but also that the recursion yields a unique limit.
Thus $x^\iteration$ has the same limit as $\interimtrust{\iteration}{}$,
which is $\trust{}$.
\end{proof}
\begin{lemma}
\label{lemma:bound_first_difference}
    $\norm{x^1 - x^0}{1} \leq \beta \trust{\contributor}^{-\contributor}$.
\end{lemma}
\begin{proof}
    Denote $\vouch^{-\contributor}$ the vouch matrix obtained from $\vouch$, 
    but modified to reflect the fact that contributor $\contributor$ only vouches for the sink, i.e.
    \begin{equation}
        \vouch^{-\contributor}_{\contributorbis k} \triangleq \left\lbrace 
        \begin{array}{ll}
            \vouch_{\contributorbis k} & \text{if } \contributorbis \neq \contributor,  \\
            0 & \text{if } \contributorbis = \contributor.
        \end{array}
        \right.
    \end{equation}
    The characteristic equation verified by $\trust{}^{-\contributor}$ is then given by
    \begin{equation}
        \trust{}^{-\contributor} = \min \set{ \pretrust{} + \vouchdecay (\vouch^{-\contributor})^T \trust{}^{-\contributor}, 1 }.
    \end{equation}
    But then we can bound the difference between $x^1$ and $x^0 = \trust{}^{-\contributor}$ as (using also Proposition~\ref{prop:trust} and Lemma~\ref{lemma:min_lipschitz})
    \begin{align}
        \norm{x^1 - x^0}{1}
        &= \sum_{\contributorbis \in \Contributor} \absv{ 
            \min \set{\pretrust{\contributorbis} + \vouchdecay \sum_{k \in \Contributor} \vouch_{k \contributorbis} \trust{k}^{-\contributor}, 1} 
            - \min \set{\pretrust{\contributorbis} + \vouchdecay \sum_{k \in \Contributor} \vouch_{k \contributorbis}^{-\contributor} \trust{k}^{-\contributor}, 1} 
        } \\
        &\leq \sum_{\contributorbis \in \Contributor} \absv{ 
            \left( \pretrust{\contributorbis} + \vouchdecay \sum_{k \in \Contributor} \vouch_{k \contributorbis} \trust{k}^{-\contributor} \right)
            - \left( \pretrust{\contributorbis} + \vouchdecay \sum_{k \in \Contributor} \vouch_{k \contributorbis}^{-\contributor} \trust{k}^{-\contributor} \right)
        } \\
        &= \sum_{\contributorbis \in \Contributor} \absv{ 
            \vouchdecay \vouch_{\contributor \contributorbis} \trust{\contributor}^{-\contributor}
        } 
        = \vouchdecay \trust{\contributor}^{-\contributor} \sum_{\contributorbis \in \Contributor} \vouch_{\contributor \contributorbis} 
        \leq \vouchdecay \trust{\contributor}^{-\contributor},
    \end{align}
    which is the lemma.
\end{proof}
\begin{reptheorem}{th:trust_security}
Denoting $\trust{}^{-\contributor}$ the trusts computed after removing $\contributor$'s vouches (and thus having $\contributor$ vouching only for the sink), 
\begin{equation}
    \sum_{\contributorbis \in \Contributor} \absv{\trust{\contributorbis} - \trust{\contributorbis}^{-\contributor}} \leq \frac{\vouchdecay}{1-\vouchdecay} \trust{\contributor}^{-\contributor}.
\end{equation}
\end{reptheorem}
\begin{proof}
    Using Lemma~\ref{lemma:bound_first_difference}, $\vouchdecay$-contractiveness and induction, 
    it follows that $\norm{x^{\iteration + 1} - x^\iteration}{1} \leq \vouchdecay^\iteration \norm{x^1 - x^0}{1} \leq \vouchdecay^{\iteration + 1} \trust{\contributor}^{-\contributor}$.
    But then,
    \begin{equation}
        \norm{x^{\iteration + 1} - x^0}{1}
        \leq \sum_{s = 0}^\iteration \norm{x^{s+1} - x^s}{1}
        \leq \sum_{s = 0}^\iteration \vouchdecay^{s + 1} \trust{\contributor}^{-\contributor}
        \leq \trust{\contributor}^{-\contributor} \sum_{s = 0}^\infty \vouchdecay^{s + 1}
        = \frac{\vouchdecay}{1-\vouchdecay} \trust{\contributor}^{-\contributor}.
    \end{equation}
    Taking this inequality to the limit $\iteration \rightarrow \infty$ and invoking Lemma~\ref{lemma:x_convergence} yields the theorem.
\end{proof}
\subsection{Other propositions}
\label{app:lipschitrust_propositions}
\begin{repproposition}{prop:byztrust_monotone}
  For all contributors $\contributor$, $\trust{} \geq \trust{}^{-\contributor}$.
\end{repproposition}
\begin{proof}
    We consider the same notations as in the proof of Lemma~\ref{lemma:bound_first_difference}.
    Since $\vouch \geq \vouch^{-\contributor}$, it is clear that $x^1 \geq x^0$.
    But then, by induction, given $\vouch \geq \vouch^{-\contributor}$ and $x^t \geq x^0$,
    it follows from the recursion of $x^t$ and the characteristic equation of $\trust{}^{-\contributor}$ that $x^{t+1} \geq x^0$.
    Taking this to the limit yields the proposition.
\end{proof}
\begin{repproposition}{prop:positive_trust}
    A contributor has a positive trust score, if and only if, there is a vouch path from a pretrusted contributor to them.
\end{repproposition}
\begin{proof}
    Consider a sequence $\contributor_0, \contributor_1, \ldots, \contributor_k$ 
    such that $\contributor_0$ is pretrusted and 
    such that $\contributor_s$ has vouched for $\contributor_{s+1}$ for all $s \in \set{0, 1, \ldots, k-1}$.
    By induction over $s$, let us prove that, for all $\iteration \geq s$, $\interimtrust{\iteration}{\contributor_s} \geq \vouchdecay^s \pretrust{\accept} \prod_{r=0}^{s-1} \vouch_{\contributor_{r} \contributor_{r+1}}$.
    \begin{itemize}
        \item For $s = 0$, this trivially follows from the fact that $\interimtrust{\iteration}{\contributor_0} 
        \geq \min \set{\pretrust{\contributor_0}, 1} 
        = \min \set{\pretrust{\accept}, 1} 
        = \pretrust{\accept}$.
        \item Now assume that it holds for $s$, and let $\iteration \geq s+1$.
        Then $\interimtrust{\iteration}{\contributor_{s+1}} \geq \min \set{ \vouchdecay \vouch_{\contributor_s \contributor_{s+1}} \interimtrust{\iteration - 1}{s}, 1 }$.
        Leveraging our induction assumption allows to close the inductive proof.
    \end{itemize}
    Now taking the induction guarantee to the limit $\iteration \rightarrow \infty$ implies $\trust{\contributor_s}{} \geq \vouchdecay^s \pretrust{\accept} \prod_{r=0}^{s-1} \vouch_{\contributor_{r} \contributor_{r+1}} > 0$.
    Thus, if there is a vouch path from a pretrusted contributor to a contributor, then this latter contributor has a positive trust score.
    Conversely, let us prove that, in the absence of such a vouch path, the contributor has a zero trust score.
    More precisely, define $\Contributor_{\reject \reject}$ the set of contributors that are not pretrusted, and for which no path exists from pretrusted contributors to them.
    We prove by induction over $\iteration$ that $\forall \contributor \in \Contributor_{\reject \reject} \mathsep \interimtrust{\iteration}{\contributor} = 0$.
    \begin{itemize}
        \item For $\iteration = 0$, this trivially follows from the fact that contributors in $\Contributor_{\times \times}$ are not pretrusted.
        \item Assume it holds for iteration $\iteration$. 
        Note that all contributors $\contributor \in \Contributor_{\reject \reject}$ are necessarily only vouched by other contributors in $\Contributor_{\reject \reject}$, since, otherwise, there would be a path from a pretrusted contributor to them.
        As a result, if $\vouch_{\contributorbis \contributor} > 0$, then $\contributorbis \in \Contributor_{\reject \reject}$, 
        and thus, by induction $\interimtrust{\iteration}{\contributorbis} = 0$.
        But then, $(\vouch^T)_\contributor \interimtrust{\iteration}{} = 0$.
        Since $\pretrust{\contributor} = 0$, computing the recursion then yields $\interimtrust{\iteration + 1}{\contributor} = 0$.
        This concludes the induction.
    \end{itemize}
    Taking $\forall \contributor \in \Contributor_{\reject \reject} \mathsep \interimtrust{\iteration}{\contributor} = 0$ then yields the other implication of the proposition.
\end{proof}
\section{Generalized Bradley-Terry model}
\label{app:betabt}
We rely on the generalized Bradley-Terry (GBT) model that has been recently introduced in~\cite{GBT2023}. 
In this model, the conditional law of the comparisons $\normalizedcomparison_{\alternative\alternativebis}$ with respect to the score vector $\contributorscore$ is of the form
\begin{equation}
 \pdf{\normalizedcomparison_{\alternative\alternativebis} | \contributorscore} = \frac{ \exp( - \contributorscore_{\alternative \alternativebis} \normalizedcomparison_{\alternative \alternativebis})
    \pdf{\normalizedcomparison_{\alternative \alternativebis} |  0}}
    {\int_{-1}^1  \exp(- \contributorscore_{\alternative \alternativebis} s)
    \pdf{s | 0} \mathrm{d}s}
\end{equation}
for some fixed probability law $\pdf{\normalizedcomparison|0}$ which corresponds to the probability law of a comparison for null scores.
The denominator ensures that $ \pdf{\normalizedcomparison_{\alternative\alternativebis} | \contributorscore}$ is a pdf with integral $1$. 
The function $\contributorscore \mapsto \int_{\mathbb{R}} \exp(\contributorscore s)
    \pdf{s | 0} \mathrm{d}s$ is the moment generating function of $\normalizedcomparison_{\alternative\alternativebis} | \contributorscore_{\alternative\alternativebis} = 0$. Its logarithm is the cumulant generating function given by
    $\Phi(\contributorscore) = \log  ( \int_{-1}^1 \exp(- \contributorscore s)
    \pdf{s | 0} \mathrm{d}s  )$. 
    The cumulant generating function  $\Phi$ is known to be strictly convex~\cite{GBT2023}. Its expression is moreover known for many probability law $ \pdf{s | 0}$. 
    In this paper, we consider the case where $ \pdf{s | 0} = 1/2$, i.e. the uniform law. This corresponds to 
    \begin{equation}
                \Phi(\contributorscore) = \log \left( \frac{1}{2} \int_{-1}^1 \exp( - \contributorscore s)
\mathrm{d}s \right) = \log \frac{\sinh(\theta)}{\theta}.
    \end{equation}
    Using this relation, we recover the expression~\eqref{eq:contributorloss}. 
\paragraph{The spring analogy}
To provide more insights into our model, a physical analogy may be useful.
Interestingly, the loss of Bayesian GBT models corresponds precisely to the total energy of a physical system, where each alternative is a pebble put in a parabolic valley, and where (non-Hookean) springs of length at rest $t^*_{\alternative \alternativebis}$ (which verifies $\normalizedcomparison_{\alternative \alternativebis} = \frac{1}{t^*_{\alternative \alternativebis}} - \coth{} (t^*_{\alternative \alternativebis})$) are used to attach alternatives $\alternative$ to $\alternativebis$.
A pebblenet would then be obtained, where pebbles naturally fall towards the bottom of the valley, but where springs may push them away from the bottom.
If pebbles are then forced to lie along a one-dimensional line of the valley, then their equilibrium position along this line would correspond to the scores assigned to the alternatives when minimizing the contributor's loss function $\contributorloss{} (\contributorscore | \normalizedcomparison)$.
\section{Primitives and implementation details}
\label{app:primitives}
To satisfy provable security guarantees, 
our algorithms repeatedly leverage a set of key Lipschitz-resilient primitives, 
many of which are adapted from~\cite{AllouahGHV22}.
In particular, we generalize them to account for \emph{asymmetric uncertainties} on the inputs.
But first, let us recall what is precisely meant by $\lipschitz$-Lipschitz resilience.
\begin{definition}[\cite{AllouahGHV22}]
    Consider a function $f : \setR_+^\Contributor \times Z^\Contributor \rightarrow X$ 
    that maps a pair $(\votingright, z)$ of contributors' voting rights $\votingright_\contributor$ and of their inputs $z_n$
    to outputs $f(\votingright, z)$ that lie in a metric space $(X, d_X)$.
    We say that $f$ is $(\lipschitz, d_X)$-Lipschitz resilient 
    if it is $\lipschitz$-Lipschitz continuous in voting rights, under $\ell_1$ norm for the voting rights, i.e.
    \begin{equation}
        \forall \votingright, \votingright' \in \setR_+^\Contributor \mathsep
        \forall z \in Z^\Contributor \mathsep
        d_X(f(\votingright, z), f(\votingright', z)) 
        \leq \lipschitz \norm{w - w'}{1}.
    \end{equation}
\end{definition}
In the context of Tournesol, $z_\contributor$ is tyically the comparison tensor $\comparison_\contributor$ reported by the contributor.
But it may also be quantities derived from the tensor, such as scores with uncertainties.
Note that this definition clearly implies that the maximal influence of a contributor $\contributor$
is proportional to its voting right.
Indeed, measuring this influence amounts to comparing the effective voting rights $\votingright$
to the ones obtained by only canceling the contributor $\contributor$'s voting right,
i.e. $\votingright'_\contributor \triangleq 0$ and $\votingright'_\contributorbis \triangleq \votingright_\contributorbis$ for $\contributorbis \neq \contributor$.
\subsection{Asymmetric Huber loss}
Note that the asymmetric Huber loss is continuous and twice continuously differentiable in $\setR$, 
with a derivative given by
\begin{equation}
    \frac{d}{dm} \huber_\alpha (m|x, \uncertainty) \triangleq \left\lbrace
    \begin{array}{ll}
        \min \set{ 1, \frac{1-\alpha}{\alpha}} \frac{m-x}{ \sqrt{\uncertainty_{left}^2 + (x-m)^2 } } & \text{for } m \leq x,  \\
        \min \set{ 1, \frac{\alpha}{1-\alpha}} \frac{m-x}{ \sqrt{\uncertainty_{right}^2 + (x-m)^2 } } & \text{for } m \geq x.
    \end{array}
    \right.
\end{equation}
It is also easy to see that the second derivative is nonnegative, 
which proves that the asymmetric Huber loss is convex.
In particular, its unique minimum is reached at $m=x$.
\subsection{The quadratically regularized median \qrmed{}~\cite{AllouahGHV22}}
\paragraph{Definition.}
The key primitive used by \mehestan{} is the \emph{quadratically regularized median} \qrmed{},
which we generalize to account for asymmetric uncertainties. 
Namely, given voting rights $\votingright \in \setR_+^\Contributor$, 
partial scores $x \in (\setR \cup \set{\perp})^\Contributor$ 
and partial score left and right uncertainties $\uncertainty_{left}, \uncertainty_{right} \in (\setR \cup \set{\perp})^\Contributor$,
we define their $\lipschitz$-quadratically regularized median by
\begin{equation}
\label{eq:qrmed}
  \qrmed{}_\lipschitz (\votingright, x, \uncertainty)
  \triangleq \argmin{m \in \setR}{
    \qrmedloss{\lipschitz} (m | \votingright, x, \uncertainty) \triangleq \frac{m^2}{2 \lipschitz} + \sum_{\contributor : x_\contributor \neq \perp} \votingright_\contributor \huber_{1/2}(m|x_\contributor, \uncertainty_\contributor)
  }.
\end{equation}
This operator can be made to ignore some contributors.
More precisely, if we want to restrict it to the subset
$\EmailVerifiedContributor$ of email-verified contributors,
then we may write it as
$\qrmed{}_\lipschitz (\votingright_\contributor, x_\contributor, \uncertainty_\contributor | \contributor \in \EmailVerifiedContributor)$.
The key safety property of \qrmed{} is $\lipschitz$-Byzantine resilience, which we redefine here.
\begin{proposition}[adapted from \cite{AllouahGHV22}]
\label{prop:qrmed}
  $\qrmed{}_\lipschitz$ is well-defined, $\lipschitz$-Lipschitz resilient and $1$-Lipschitz continuous in $x$ with respect to the $\ell_\infty$ norm.
\end{proposition}
\begin{proof}
    Note that~\cite{AllouahGHV22} discusses \qrmed{} without uncertainty, which essentially corresponds to $\Delta_\contributor \rightarrow 0$.
    However, their proof relies on controlling the derivative of $\qrmedloss{\lipschitz}' (m | \votingright, x, \uncertainty)$,
    and in particular on the fact that, if we focus only the term of $\qrmedloss{\lipschitz}' (m | \votingright, x, \uncertainty)$ that corresponds to a contributor $\contributor$,
    then the derivative is at most $\votingright_\contributor$ in absolute value.
    These key properties still hold for our generalization of $\qrmed{}$.
    Thus their proof applies to our setting as well.
\end{proof}
\qrmed{} has another desirable property.
Namely, for any $\varepsilon > 0$, when the number of contributors with unit voting right is large enough (especially compared to $\lipschitz$), then $\qrmed{}$ returns a value that lies between the quantiles $1/2-\varepsilon$ and $1/2+\varepsilon$.
In the limit of a very large number of contributors, \qrmed{} thus behaves essentially like a median.
\paragraph{Computation.}
Interestingly, \qrmed{} can be provably efficiently computed through dichotomic search.
In practice, however, we run the Brent-Dekker method~\cite{dekker69,brent71}, which provides better empirical performances,
despite the lack of theoretical guarantee,
by leveraging the smoothness of $\qrmedloss{}$.
Below we briefly describe this method, which relies on updating a lower bound $\underline{m}$ and an upper bound $\overline{m}$,
by leveraging the value of the derivative of the loss of Equation~\eqref{eq:qrmed}.
This derivative is given by
\begin{equation}
\label{eq:qrmed_derivative}
  \qrmedloss{\lipschitz}'(m | \votingright, x, \uncertainty) =  \frac{m}{\lipschitz} + \sum_{\contributor : x_\contributor \neq \perp} \votingright_\contributor \frac{ m - x_\contributor }{\sqrt{ \uncertainty_{\contributor, left}^2 {\bf 1}[m \leq x] + \uncertainty_{\contributor, right}^2 {\bf 1}[m>x] + ( x_\contributor - m )^2 }}.
\end{equation}
The Brent-Dekker algorithm is implemented in the python library\footnote{For Tournesol, we set $\varepsilon \triangleq 0.01$.} \texttt{scipy.optimize.brentq}, which we adapted (so that we can run it with \texttt{numba} for better performances).
\begin{proposition}
  The initialization step takes $\bigO \left(\log \left( 1+\absv{\qrmed{}_\lipschitz (\votingright, x, \uncertainty)} \right) \right)$ iterations.
  The dichotomic step takes $\bigO (\log (1/\varepsilon))$ iterations.
  In both cases, each iteration requires computing $\qrmedloss{\lipschitz}'$, which requires $\bigO(\card{\Contributor})$ steps.
\end{proposition}
\paragraph{Subsampling for faster computation.}
Note that the computation of $\qrmedloss{\lipschitz}'$ may be costly, if there is a very large number of contributors with well-defined imputs.
If so, then a random sampling of the contributors may be used instead, especially for early steps of the initialization and dichotomic search.
In, investigating how subsampling should be done and what theoretical guarantees it provides is however left for future work.
\subsection{The Byzantine-robustified mean \brmean{}~\cite{AllouahGHV22}}
We will use another robust statistics primitive,
which successfully returns the mean of bounded inputs,
when the amount of allocated voting rights is large enough.
We call it \brmean{}, for \emph{Byzantine-Robustified Mean}.
To define it, we first introduce the \emph{Clipped Mean} $\clippedmean{}$ centered on $\clippedmeancenter$ and of radius $\clippedmeanradius$ by
\begin{align}
    \clippedmean{} (\votingright, x, \uncertainty | \clippedmeancenter, \clippedmeanradius)
    &\triangleq \mean (\votingright, \clip(x | \clippedmeancenter, \clippedmeanradius)) \\
    &= \frac{1}{\sum_{\contributor \in \Contributor} \votingright_\contributor}
    \sum_{\contributor \in \Contributor} \votingright_{\contributor}
    \clip (x_\contributor | \clippedmeancenter, \clippedmeanradius)
\end{align}
where $\clip( x | \clippedmeancenter, \clippedmeanradius) \triangleq \max \set{ \clippedmeancenter - \clippedmeanradius, \min \set{ \clippedmeancenter + \clippedmeanradius, x } }$ clips $x$ within the interval $[\clippedmeancenter - \clippedmeanradius, \clippedmeancenter + \clippedmeanradius]$.
Note that, while we made it appear explicitly to conform with other primitives, 
the uncertainty $\uncertainty$ is never used.
\brmean{} is then obtained by executing $\clippedmean{}$, centered on $\qrmed{}$, with a radius that grows linearly with the total amounts of votes:
\begin{equation}
    \brmean_{\lipschitz} (\votingright, x, \uncertainty )
    \triangleq \clippedmean{} \left( \votingright, x, \uncertainty \st \qrmed_{4 \lipschitz} (\votingright, x, \uncertainty ), \frac{\lipschitz}{4} \sum_{\contributor \in \Contributor} \votingright_\contributor \right).
\end{equation}
Crucially for our purposes, \brmean{} has the following properties.
\begin{proposition}[\cite{AllouahGHV22}]
\label{prop:brmean}
$\brmean_{\lipschitz}$ is $\lipschitz$-Byzantine resilient.
Moreover, if there exists $\clippedmeanradius>0$ such that $\lipschitz \sum_{\contributor \in \Contributor} \votingright_\contributor \geq 8 \clippedmeanradius$ and $x_\contributor \in [-\clippedmeanradius, \clippedmeanradius]$ for all $\contributor \in \Contributor$,
then $\brmean_{\lipschitz} (\votingright, x, \uncertainty) = \mean{} (\votingright, x)$.
\end{proposition}
We stress, however, that \brmean{} is oblivious to uncertainty.
It should thus only be used when there are enough data, so that the uncertainty can be neglected.
\subsection{The quadratically regularized quantile \qrqtl{}}
\label{app:qr_quantile}
We now introduce a new estimator, called the quadratically regularized quantile \qrqtl{},
which depends on a quantile $\alpha$ and on a Lipschitz resilience $\lipschitz$.
We define it by
\begin{equation}
\label{eq:qrqtl2}
  \qrqtl{}_{\alpha, \lipschitz} (\votingright, x, \uncertainty)
  \triangleq \argmin{m \in \setR}{
    \frac{m^2}{2 \lipschitz} 
    + \sum_{\contributor : x_\contributor \neq \perp} 
        \votingright_\contributor
        \huber_\alpha (m|x_\contributor, \uncertainty_\contributor)
  }.
\end{equation}
\begin{proposition}
    $\qrqtl{}_{\alpha, \lipschitz}$ is $\lipschitz$-Lipschitz resilient.
\end{proposition}
\begin{proof}
    For any $x, \uncertainty$, the loss 
    $\ell(m | x, \uncertainty) \triangleq \huber(m|x, \uncertainty)$ has a derivative
    $\ell'(m | x, \uncertainty) = \frac{d}{dm} \huber(m|x, \uncertainty)$,
    whose absolute value can be bounded by
    $\absv{\ell'(m | x, \uncertainty)} \leq 1$.
    Fix $\votingright_1, \votingright_2 \in \setR_+^\Contributor$,
    and denote $q_1 \triangleq \qrqtl_{\alpha, \lipschitz} (\votingright_1, x, \uncertainty)$
    and $q_2 \triangleq \qrqtl_{\alpha, \lipschitz} (\votingright_2, x, \uncertainty)$.
    By optimality of \qrqtl{},
    we know that the derivative of the loss must cancel, which implies
    \begin{align}
        q_1 = - \lipschitz \sum_\contributor \votingright_{1,\contributor} \ell'(q | x_\contributor, \uncertainty_\contributor) \quad \text{and} \quad
        q_2 = - \lipschitz \sum_\contributor \votingright_{2, \contributor} \ell'(q' | x_\contributor, \uncertainty_\contributor).
    \end{align}
    Taking the difference and bounding its absolute value then yields
    \begin{align}
        \absv{q_1 - q_2} 
        \leq \lipschitz \sum_\contributor \absv{\votingright_{1,\contributor} - \votingright_{2,\contributor}} \absv{\ell'(q' | x_\contributor, \uncertainty_\contributor)} 
        \leq \lipschitz \norm{\votingright_1 - \votingright_2}{1}.
    \end{align}
    This is precisely the needed Lipschitz resilience guarantee.
\end{proof}
\begin{proposition}
    Assume $\uncertainty = 0$ and denote $q_\beta$ a weighted $\beta$-quantile of $(\votingright, x)$.
    If $\beta \leq \alpha - \absv{q_\beta} / 2 \lipschitz \norm{\votingright}{1}$,
    then $\qrqtl{}_{\alpha, \lipschitz } (\votingright, x, \uncertainty) \geq q_\beta$.
    Similarly, if $\beta \geq \alpha + \absv{q_\beta} / 2 \lipschitz \norm{\votingright}{1}$,
    then $\qrqtl{}_{\alpha, \lipschitz} (\votingright, x, \uncertainty) \leq q_\beta$.
    In particular, in the limit $L \rightarrow \infty$ (no resilience)
    or $\norm{w}{1} \rightarrow \infty$ (high participation), 
    $\qrqtl_{\alpha + o(1), L}$ becomes a great estimate of $q_\beta$.
\end{proposition}
\begin{proof}
    Define $\Contributor_<(q)$, $\Contributor_=(q)$ and $\Contributor_>(q)$ the subsets of contributors $\contributor$
    whose scores $x_\contributor$ are respectively smaller, equal and larger than $q$.
    Denote $\votingright_?(q) \triangleq \sum_{\contributor \in \Contributor_?(q)} \votingright_\contributor$
    for $? \in \set{<, =, >}$.
    By definition of the weighted $\beta$-quantile, we must have
    \begin{align}
        \label{eq:q_beta1}
        \votingright_<(q_\beta) \leq \beta \norm{\votingright}{1}
        \quad &\text{and} \quad 
        \votingright_<(q_\beta) + \votingright_=(q_\beta) \geq \beta \norm{\votingright}{1}, \\        
        \label{eq:q_beta2}
        \votingright_>(q_\beta) \leq (1-\beta) \norm{\votingright}{1}
        \quad &\text{and} \quad 
        \votingright_>(q_\beta) + \votingright_=(q_\beta) \geq (1-\beta) \norm{\votingright}{1}.
    \end{align}
    Now assume $\alpha \leq 0.5$. 
    The subgradient $\nabla_{q_\beta}$ of the loss of \qrqtl{} at $q_\beta$ is given by
    \begin{align}
        \nabla_{q_\beta} &= \frac{q_\beta}{\lipschitz} 
        - \votingright_< (q_\beta)  
        + \votingright_= (q_\beta) \left[ - 1, \frac{\alpha}{1 - \alpha} \right]
        + \frac{\alpha}{1 - \alpha} \votingright_> (q_\beta) \\
    \end{align}
    Using the right equation of~\eqref{eq:q_beta1} and the left equation of~\eqref{eq:q_beta2},
    we have
    \begin{align}
        (1-\alpha) \min \nabla_{q_\beta} 
        &= (1-\alpha) \frac{q_\beta}{\lipschitz} - (1-\alpha) \votingright_<(q_\beta) - (1-\alpha) \votingright_=(q_\beta) + \alpha \votingright_>(q_\beta) \\
        &\leq (1-\alpha) \frac{q_\beta}{\lipschitz} - (1 - \alpha) \beta \norm{w}{1} + \alpha (1- \beta) \norm{w}{1} \\
        &= (1-\alpha) \frac{q_\beta}{\lipschitz} + (\alpha- \beta) \norm{w}{1}.
    \end{align}
    Assuming $\beta \leq \alpha - \absv{q_\beta} / 2 \lipschitz \norm{\votingright}{1}$
    implies $\beta \leq \alpha - (1-\alpha) \absv{q_\beta} / 2 \lipschitz \norm{\votingright}{1}$,
    we have the guarantee $\min \nabla_{q_\beta} \leq 0$,
    which implies that the minimum of \qrqtl{} is to the right of $q_\beta$.
    Considering now the other side, we have
    \begin{align}
        (1-\alpha) \max \nabla_{q_\beta} 
        &= (1-\alpha) \frac{q_\beta}{\lipschitz} - (1-\alpha) \votingright_<(q_\beta) + \alpha \votingright_=(q_\beta) + \alpha \votingright_>(q_\beta) \\
        &\geq (1-\alpha) \frac{q_\beta}{\lipschitz} - (1 - \alpha) \beta \norm{w}{1} + \alpha (1- \beta) \norm{w}{1} \\
        &= (1-\alpha) \frac{q_\beta}{\lipschitz} + (\alpha- \beta) \norm{w}{1}.
    \end{align}
    Like earlier, in the case $\beta \geq \alpha + \absv{q_\beta} / 2 \lipschitz \norm{\votingright}{1}$,
    we have $\max \nabla_{q_\beta} \geq 0$, thereby implying that the minimum is to the left of $q_\beta$.
    The case $\alpha \geq 0.5$ is dealt with similarly,
    using $\alpha \nabla_{q_\beta} = \alpha \frac{q_\beta}{L} - (1-\alpha) w_<(q_\beta)
    + w_=(q_\beta) \left[ - (1-\alpha), \alpha \right] + \alpha w_>(q_\beta)$
    in this setting.
\end{proof}
\begin{proposition}
\label{prop:infinite_uncertainty}
    Let $\contributor \in [\Contributor]$.
    For any $\alpha \in (0, 1)$, $L \geq 0$ and $w, x, \Delta$, we have
    \begin{equation}
        \lim_{\Delta_{\contributor, left}', \Delta'_{\contributor, right} \rightarrow \infty}
        \qrqtl_{\alpha, L} ( w, x, (\Delta_\contributor', \Delta_{-\contributor})) =
        \qrqtl_{\alpha, L} ( (0, w_{-\contributor}), x, \Delta).
    \end{equation}
    Put differently, a user with extremely large left and right uncertainties
    is asymptotically equivalent to a user with zero voting right.
\end{proposition}
\begin{proof}
    Let $q = \qrqtl_{\alpha, L} ( (0, w_{-\contributor}), x, \Delta)$ and $\varepsilon > 0$.    
    Now, since the loss of \qrqtl{} is a sum 
    of a $\frac{1}{L}$-strongly convex term (the regularization)
    and of other convex terms,
    it is itself $\frac{1}{L}$-strongly convex.
    Thus for $q'$ such that $\absv{q' - q} \geq \varepsilon$,
    the derivative of \qrqtl{} for inputs $(0, w_{-\contributor}), x, \Delta$ at $q'$
    is, in absolute value, at least $\varepsilon / L$.
    Now observe that, on input $w, x, \Delta$,
    this derivative is the same as above, 
    plus a term $w_u \frac{d}{dq} \huber_\alpha (q' | x_u, \Delta_u)$ 
    for user $u$.
    Note however that, on $[q - \varepsilon, q + \varepsilon]$ (in fact, on any interval),
    the term $w_u \frac{d}{dq} \huber_\alpha (q' | x_u, \Delta_u')$ converges uniformly to zero
    as $\Delta_{u, left}', \Delta_{u, right}' \rightarrow \infty$.
    More precisely, there exists $A$ such that,
    if $\Delta_{u, left}' \geq A$ and $\Delta_{u, right}' \geq A$,
    we have $\absv{w_u \frac{d}{dq} \huber_\alpha (q' | x_u, \Delta_u')} \leq \varepsilon / 2L$
    for all $q' \in [q - \varepsilon, q + \varepsilon]$.
    This then implies that the derivative of the \qrqtl{} loss 
    for inputs $w, x, (\Delta_\contributor', \Delta_{-\contributor})$ 
    at $q + \varepsilon$ will be at least $\frac{\varepsilon}{L} - \frac{\varepsilon}{2L} > 0$.
    Thus the minimum is on the left of $q+\varepsilon$.
    Similarly, we can prove that the minimum must be on the right of $q-\varepsilon$.
    Thus the \qrqtl{} must be $[q-\varepsilon, q+\varepsilon]$,
    which concludes the proof.
\end{proof}
\subsection{The quadratically regularized deviation \qrdev{}}
We now propose a Lipschitz-resilient measure of standard deviation, which we call the \emph{Quadratically Regularized Deviation}, or $\qrdev{}$.
We essentially define it as a \qrqtl{} of the deviation to \qrmed{}.
More specifically,
given the quadratically regularized median $m \triangleq \qrmed{}_\lipschitz (\votingright, x, \uncertainty)$ and a default deviation $\defaultdeviation$,
$\qrdev{}$ is computed by
\begin{equation}
  \qrdev{}_{\alpha, \lipschitz, \defaultdeviation} (\votingright, x, \uncertainty) \triangleq
  \defaultdeviation + \qrqtl{}_{\alpha, \lipschitz} \left(
    \votingright_\contributor,
    \absv{x_\contributor - m} - \defaultdeviation,
    \uncertainty_\contributor
    \st \contributor \in \Contributor
  \right).
\end{equation}
Note that, rather than a measure of deviation within the values of $x_\contributor$,
$\qrdev{}$ should rather be regarded as a deviation between these values and $\qrmed{}$.
It may also be understood as a measure of polarization.
More precisely, it would take large values if most contributors disagree with the quadratically regularized median.
\begin{proposition}
  $\qrdev{}_{\alpha, \lipschitz, \uncertainty_{default}}$ is $2 \lipschitz$-Lipschitz resilient.
\end{proposition}
\begin{proof}
  Essentially, the proposition results from the Lipschitz resilience of \qrmed{} and \qrqtl{},
  and the 1-Lipchitz continuity of \qrqtl{} with respect to a translation of its second inputs.
  Let $w,w' \in \setR_+^\Contributor$.
  By Proposition~\ref{prop:qrmed}, denoting $m \triangleq \qrmed{}_\lipschitz (\votingright, x, \uncertainty)$ and $m' \triangleq \qrmed{}_\lipschitz (\votingright', x, \uncertainty)$,
  we know that $\absv{m - m'} \leq \norm{w - w'}{1} / \lipschitz$.
  Denoting $d$ and $d'$ the $\qrdev{}$ for $w$ and $w'$, we have
  \begin{align}
      \absv{d - d'} 
      &\leq \absv{ \qrqtl{}_{\alpha, \lipschitz} ( \votingright_\contributor, \absv{x_\contributor - m} - \defaultdeviation, \uncertainty_\contributor) - \qrqtl{}_{\alpha, \lipschitz} ( \votingright_\contributor, \absv{x_\contributor - m'} - \defaultdeviation, \uncertainty_\contributor)} \nonumber \\
      &+ \absv{ \qrqtl{}_{\alpha, \lipschitz} ( \votingright_\contributor, \absv{x_\contributor - m'} - \defaultdeviation, \uncertainty_\contributor) - \qrqtl{}_{\alpha, \lipschitz} ( \votingright'_\contributor, \absv{x_\contributor - m'} - \defaultdeviation, \uncertainty_\contributor)} \\
      &\leq \absv{m - m'} + \lipschitz \norm{w - w'}{1}
      \leq 2 \lipschitz \norm{w - w'}{1},
  \end{align}
  which amounts to saying that $\qrdev{}_\lipschitz$ is $(2 \lipschitz)$-Byzantine resilient.
\end{proof}
\subsection{The quadratically regularized uncertainty}
The quadratically regularized uncertainty aims to provide a meaningful estimate on the estimation of $\qrmed{}$.
In our default pipeline the quadratically regularized uncertainty is simply made to return $\qrdev{}$.
This is a pessimistic assumption, if the data is i.i.d.
However, given that we do not expect i.i.d. data, we considered this more conservative estimator 
of the uncertainty on \qrmed{}.
We leave the problem of deriving more principled uncertainty estimates open.
\section{Mehestan's algorithms}
\label{app:mehestan}
This section details the algorithms used by Mehestan.
\subsection{Collaborative preference scaling for scaling-calibration contributors}
For any contributor $\contributor \in \Contributor$, 
define the set of clearly ordered pairs of alternatives by $\contributor$, by
\begin{equation}
\label{eq:alternativepair_nonequal}
  \AlternativePair_{\contributor}^* \triangleq
  \set{
    (\alternative, \alternativebis) \in \AlternativePair_{\contributor}
    \st 
    \optcontributorscore_{\contributor \alternative}
    \geq \optcontributorscore_{\contributor \alternativebis}
    + 2 \optcontributorscoreuncertainty_{\contributor \alternative, left} 
    + 2 \optcontributorscoreuncertainty_{\contributor \alternativebis, right}
    \text{ or }
    \optcontributorscore_{\contributor \alternativebis}
    \geq \optcontributorscore_{\contributor \alternative}
    + 2 \optcontributorscoreuncertainty_{\contributor \alternative, right} 
    + 2 \optcontributorscoreuncertainty_{\contributor \alternativebis, left}
  }.
\end{equation}
We then say that two contributors $\contributor$ and $\contributorbis$ are \emph{comparable},
if and only if, there exists at least one pair of alternatives that both scored sufficiently differently,
i.e. $\AlternativePair_{\contributor}^* \cap \AlternativePair_{\contributorbis}^* \neq \emptyset$.
For any two comparable contributors $\contributor$ and $\contributorbis$,
each pair $(\alternative, \alternativebis) \in \AlternativePair_{\contributor}^* \cap \AlternativePair_{\contributorbis}^*$ may then serve as an anchor point to appropriately scale the two contributors' scores.
This leads us to define
$\scaling_{\contributor \contributorbis \alternative \alternativebis} \triangleq  \frac{\absv{ \optcontributorscore_{\contributorbis \alternative} - \optcontributorscore_{\contributorbis \alternativebis} }}{\absv{ \optcontributorscore_{\contributor \alternative} - \optcontributorscore_{\contributor \alternativebis} }}$.
We also compute their left uncertainty as
\begin{align}
   \scalinguncertainty_{\contributor \contributorbis \alternative \alternativebis, left} \triangleq \left\lbrace
   \begin{array}{cc}
      \scaling_{\contributor \contributorbis \alternative \alternativebis}  -
      \frac{
          \optcontributorscore_{\contributorbis \alternative} - \optcontributorscore_{\contributorbis \alternativebis}
          - \optcontributorscoreuncertainty_{\contributorbis \alternative, left} 
          - \optcontributorscoreuncertainty_{\contributorbis \alternativebis, right} 
      }{
          \optcontributorscore_{\contributor \alternative} - \optcontributorscore_{\contributor \alternativebis}
          + \optcontributorscoreuncertainty_{\contributor \alternative, right} 
          + \optcontributorscoreuncertainty_{\contributor \alternativebis, left}
      }
      & \text{ if } \optcontributorscore_{\contributorbis \alternative} \geq \optcontributorscore_{\contributorbis \alternativebis} 
      \text{ and } \optcontributorscore_{\contributor \alternative} \geq \optcontributorscore_{\contributor \alternativebis}, \\
      \scaling_{\contributor \contributorbis \alternative \alternativebis}  -
      \frac{
          \optcontributorscore_{\contributorbis \alternativebis} - \optcontributorscore_{\contributorbis \alternative}
          - \optcontributorscoreuncertainty_{\contributorbis \alternative, right} 
          - \optcontributorscoreuncertainty_{\contributorbis \alternativebis, left} 
      }{
          \optcontributorscore_{\contributor \alternative} - \optcontributorscore_{\contributor \alternativebis}
          + \optcontributorscoreuncertainty_{\contributor \alternative, right} 
          + \optcontributorscoreuncertainty_{\contributor \alternativebis, left}
      }
      & \text{ if } \optcontributorscore_{\contributorbis \alternative} \leq \optcontributorscore_{\contributorbis \alternativebis} 
      \text{ and } \optcontributorscore_{\contributor \alternative} \geq \optcontributorscore_{\contributor \alternativebis}, \\
      \scaling_{\contributor \contributorbis \alternative \alternativebis}  -
      \frac{
          \optcontributorscore_{\contributorbis \alternative} - \optcontributorscore_{\contributorbis \alternativebis}
          - \optcontributorscoreuncertainty_{\contributorbis \alternative, left} 
          - \optcontributorscoreuncertainty_{\contributorbis \alternativebis, right} 
      }{
          \optcontributorscore_{\contributor \alternativebis} - \optcontributorscore_{\contributor \alternative}
          + \optcontributorscoreuncertainty_{\contributor \alternative, left} 
          + \optcontributorscoreuncertainty_{\contributor \alternativebis, right}
      }
      & \text{ if } \optcontributorscore_{\contributorbis \alternative} \geq \optcontributorscore_{\contributorbis \alternativebis} 
      \text{ and } \optcontributorscore_{\contributor \alternative} \leq \optcontributorscore_{\contributor \alternativebis}, \\
      \scaling_{\contributor \contributorbis \alternative \alternativebis}  -
      \frac{
          \optcontributorscore_{\contributorbis \alternativebis} - \optcontributorscore_{\contributorbis \alternative}
          - \optcontributorscoreuncertainty_{\contributorbis \alternative, right} 
          - \optcontributorscoreuncertainty_{\contributorbis \alternativebis, left} 
      }{
          \optcontributorscore_{\contributor \alternativebis} - \optcontributorscore_{\contributor \alternative}
          + \optcontributorscoreuncertainty_{\contributor \alternative, left} 
          + \optcontributorscoreuncertainty_{\contributor \alternativebis, right}
      }
      & \text{ if } \optcontributorscore_{\contributorbis \alternative} \leq \optcontributorscore_{\contributorbis \alternativebis} 
      \text{ and } \optcontributorscore_{\contributor \alternative} \leq \optcontributorscore_{\contributor \alternativebis}, \\
   \end{array}   
   \right.
\end{align}
and similarly for the right uncertainty:
\begin{align}
   \scalinguncertainty_{\contributor \contributorbis \alternative \alternativebis, right} \triangleq \left\lbrace
   \begin{array}{cc}
      \frac{
          \optcontributorscore_{\contributorbis \alternative} - \optcontributorscore_{\contributorbis \alternativebis}
          + \optcontributorscoreuncertainty_{\contributorbis \alternative, right} 
          + \optcontributorscoreuncertainty_{\contributorbis \alternativebis, left} 
      }{
          \optcontributorscore_{\contributor \alternative} - \optcontributorscore_{\contributor \alternativebis}
          - \optcontributorscoreuncertainty_{\contributor \alternative, left} 
          - \optcontributorscoreuncertainty_{\contributor \alternativebis, right}
      }
      - \scaling_{\contributor \contributorbis \alternative \alternativebis}
      & \text{ if } \optcontributorscore_{\contributorbis \alternative} \geq \optcontributorscore_{\contributorbis \alternativebis} 
      \text{ and } \optcontributorscore_{\contributor \alternative} \geq \optcontributorscore_{\contributor \alternativebis}, \\
      \frac{
          \optcontributorscore_{\contributorbis \alternativebis} - \optcontributorscore_{\contributorbis \alternative}
          + \optcontributorscoreuncertainty_{\contributorbis \alternative, left} 
          + \optcontributorscoreuncertainty_{\contributorbis \alternativebis, right} 
      }{
          \optcontributorscore_{\contributor \alternative} - \optcontributorscore_{\contributor \alternativebis}
          - \optcontributorscoreuncertainty_{\contributor \alternative, left} 
          - \optcontributorscoreuncertainty_{\contributor \alternativebis, right}
      }
      - \scaling_{\contributor \contributorbis \alternative \alternativebis}
      & \text{ if } \optcontributorscore_{\contributorbis \alternative} \leq \optcontributorscore_{\contributorbis \alternativebis} 
      \text{ and } \optcontributorscore_{\contributor \alternative} \geq \optcontributorscore_{\contributor \alternativebis}, \\
      \frac{
          \optcontributorscore_{\contributorbis \alternative} - \optcontributorscore_{\contributorbis \alternativebis}
          + \optcontributorscoreuncertainty_{\contributorbis \alternative, right} 
          + \optcontributorscoreuncertainty_{\contributorbis \alternativebis, left} 
      }{
          \optcontributorscore_{\contributor \alternativebis} - \optcontributorscore_{\contributor \alternative}
          - \optcontributorscoreuncertainty_{\contributor \alternative, right} 
          - \optcontributorscoreuncertainty_{\contributor \alternativebis, left}
      }
      - \scaling_{\contributor \contributorbis \alternative \alternativebis}
      & \text{ if } \optcontributorscore_{\contributorbis \alternative} \geq \optcontributorscore_{\contributorbis \alternativebis} 
      \text{ and } \optcontributorscore_{\contributor \alternative} \leq \optcontributorscore_{\contributor \alternativebis}, \\
      \frac{
          \optcontributorscore_{\contributorbis \alternativebis} - \optcontributorscore_{\contributorbis \alternative}
          + \optcontributorscoreuncertainty_{\contributorbis \alternative, left} 
          + \optcontributorscoreuncertainty_{\contributorbis \alternativebis, right} 
      }{
          \optcontributorscore_{\contributor \alternativebis} - \optcontributorscore_{\contributor \alternative}
          - \optcontributorscoreuncertainty_{\contributor \alternative, right} 
          - \optcontributorscoreuncertainty_{\contributor \alternativebis, left}
      }
      - \scaling_{\contributor \contributorbis \alternative \alternativebis}
      & \text{ if } \optcontributorscore_{\contributorbis \alternative} \leq \optcontributorscore_{\contributorbis \alternativebis} 
      \text{ and } \optcontributorscore_{\contributor \alternative} \leq \optcontributorscore_{\contributor \alternativebis}, \\
   \end{array}   
   \right.
\end{align}
Assume $\AlternativePair_{\contributor}^* \cap \AlternativePair_{\contributorbis}^* \neq \emptyset$.
Now, to mitigate noise, we estimate the relative scaling between two contributors $\contributor$ and $\contributorbis$, using a mean robust estimator.
Namely, we consider the quadratically regularized median and the quadratically regularized uncertainty estimator, which yields
\begin{equation}
\label{eq:comparative_scaling_aggregated_calibration}
  \scaling_{\contributor \contributorbis} \triangleq
  1 + \qrmed_{\lipschitz_{user}} \left(
    1,
    \scaling_{\contributor \contributorbis \alternative \alternativebis} - 1,
    \scalinguncertainty_{\contributor \contributorbis \alternative \alternativebis}
    \st (\alternative, \alternativebis) \in \AlternativePair_{\contributor}^* \cap \AlternativePair_{\contributorbis}^*
  \right).
\end{equation}
\begin{equation}
\label{eq:comparative_scaling_uncertainty_aggregated}
  \scalinguncertainty_{\contributor \contributorbis} \triangleq
  \qruncertainty_{\lipschitz_{user},1} \left(
    1,
    \scaling_{\contributor \contributorbis \alternative \alternativebis} - 1,
    \scalinguncertainty_{\contributor \contributorbis \alternative \alternativebis}
    \st (\alternative, \alternativebis) \in \AlternativePair_{\contributor}^* \cap \AlternativePair_{\contributorbis}^*
  \right).
\end{equation}
Note that these estimates aggregate the influences from a single other contributor $\contributorbis$.
This is why they only need robustness against noise,
not against malicious nodes.
Hence the use of a ``large'' Lipschitz parameter $\lipschitz_{user} = 10$.
To determine the preference re-scaling of a \emph{scaling-calibration} contributor $\contributor$,
we leverage the relative re-scaling with respect to all comparable \emph{scaling-calibration} contributors, by using the Byzantine-robustified mean, as
\begin{equation}
\label{eq:scaling_calibration}
  \scaling_{\contributor} \triangleq
  1+ \qrmed_{\lipschitz / 8 \norm{\optcontributorscore_\contributor}{\infty}} \left(
    \votingright_\contributorbis^{scaling},
    \scaling_{\contributor \contributorbis} - 1,
    \scalinguncertainty_{\contributor \contributorbis}
    \st \contributorbis \in \Contributor_{\checkmark, \contributor}^{scaling}
  \right),
\end{equation}
\begin{equation}
\label{eq:scaling_uncertainty_calibration}
  \scalinguncertainty_{\contributor} \triangleq
  \qruncertainty_{\lipschitz / 8 \norm{\optcontributorscore_\contributor}{\infty}, 1} \left(
    \votingright_\contributorbis^{scaling},
    \scaling_{\contributor \contributorbis} - 1,
    \scalinguncertainty_{\contributor \contributorbis}
    \st \contributorbis \in \Contributor_{\checkmark, \contributor}^{scaling}
  \right),
\end{equation}
where $\Contributor_{\checkmark, \contributor}^{scaling} \triangleq \set{\contributorbis \in \Contributor_{\checkmark}^{scaling} \st \AlternativePair_{\contributor}^* \cap \AlternativePair_{\contributorbis}^* \neq \emptyset}$ 
is the set of \emph{scaling-calibration} contributors comparable to $\contributor$.
Note that the contributor should include themselves, 
with $\votingright_\contributorbis^{scaling} = 1$, $\scaling_{\contributor \contributor} = 1$ and $\scalinguncertainty_{\contributor \contributor} = 0$
We define similarly the translation normalization $\translation_{\contributor \contributorbis}$ of contributor $\contributor$, compared to contributor $\contributorbis$.
To do so, let us first denote $\Alternative_{\contributor \contributorbis}
\triangleq \set{\alternative \in \Alternative \st \exists \alternativebis \in \Alternative \mathsep \comparison_{\contributor \alternative \alternativebis} \neq \perp~\text{and}~\comparison_{\contributorbis \alternative \alternativebis} \neq \perp}$
the set of alternatives that both contributors $\contributor$ and $\contributorbis$ scored.
We then define
\begin{align}
  \label{eq:comparative_translation_calibration}
  \translation_{\contributor \contributorbis} &\triangleq
  \qrmed_{1} \left(
    1,
    \scaling_\contributorbis \optcontributorscore_{\contributorbis \alternative} - \scaling_{\contributor} \optcontributorscore_{\contributor \alternative},
    \scaling_{\contributor} \optcontributorscoreuncertainty_{\contributor \alternative} + \scaling_\contributorbis \optcontributorscoreuncertainty_{\contributorbis \alternative}
    \st \alternative \in \Alternative_{\contributor \contributorbis}
  \right), \\
  \label{eq:comparative_translation_uncertainty_calibration}
  \translationuncertainty_{\contributor \contributorbis} &\triangleq
  \qruncertainty_{1,1} \left(
    1,
    \scaling_\contributorbis \optcontributorscore_{\contributorbis \alternative} - \scaling_{\contributor} \optcontributorscore_{\contributor \alternative},
    \scaling_{\contributor} \optcontributorscoreuncertainty_{\contributor \alternative} + 
    \scaling_\contributorbis \optcontributorscoreuncertainty_{\contributorbis \alternative}
    \st \alternative \in \Alternative_{\contributor \contributorbis}
  \right). 
\end{align}
To determine the preference translation of a contributor $\contributor$,
we then leverage the relative translation with respect to all comparable contributors, by defining
\begin{align}
  \label{eq:translation_calibration}
  \translation_{\contributor} &\triangleq
  \qrmed_{\lipschitz / 8} \left(
    \votingright_\contributorbis^{scaling},
    \translation_{\contributor \contributorbis},
    \translationuncertainty_{\contributor \contributorbis}
    \st \contributorbis \in \Contributor_{\checkmark, \contributor}^{scaling}
  \right) \\
  \label{eq:translation_uncertainty_calibration}
  \translationuncertainty_{\contributor} &\triangleq
  \qruncertainty_{\lipschitz / 8, 1} \left(
    \votingright_\contributorbis^{scaling},
    \translation_{\contributor \contributorbis},
    \translationuncertainty_{\contributor \contributorbis}
    \st \contributorbis \in \Contributor_{\checkmark, \contributor}^{scaling}
  \right).
\end{align}
Again each user should include themselves, with $\translation_{\contributor \contributor} = 0$ and $\translationuncertainty_{\contributor \contributor} = 0$. 
Scaling-calibration contributor $\contributor$'s \mehestan{}-scaled score is then given by
\begin{equation}
  \theta^{\mehestan{}}_{\contributor \alternative} \triangleq
  \scaling_{\contributor} \optcontributorscore_{\contributor \alternative}
  + \translation_{\contributor}.
\end{equation}
We define the uncertainty on the scaled score by
\begin{equation}
  \uncertainty \theta^{\mehestan{}}_{\contributor \alternative} \triangleq
  \scaling_{\contributor} \optcontributorscoreuncertainty_{\contributor \alternative}
  + \translationuncertainty_{\contributor}
  + \left\lbrace 
  \begin{array}{ll}
      \scalinguncertainty_{\contributor} \optcontributorscore_{\contributor \alternative} & \text{ if } \optcontributorscore_{\contributor \alternative} \geq 0 \\
      \absv{\optcontributorscore_{\contributor \alternative}} 
      (\scalinguncertainty_{\contributor, right}, \scalinguncertainty_{\contributor, left}) & \text{ if } \optcontributorscore_{\contributor \alternative} < 0
  \end{array}
  \right.
\end{equation}
\subsection{Scaling of non-scaling-calibration contributors}
Nonscalers are then made to fit the scalers' scaled scores $(\theta^{\mehestan{}}, \uncertainty \theta^{\mehestan{}})$, using similar techniques.
One slight difference occurs in equations \eqref{eq:scaling_calibration} and \eqref{eq:translation_calibration},
where $\brmean$ is used instead of $\qrmed$,
and where the scaled user excludes themselves.
\section{Hyperparameters of the synthetic evaluation}
\label{app:synthetic}
In this section we detail the hyperparameters of the experiments.
Note that they are also provided as is in the Supplementary Material,
and that experiments can be effortlessly reproduced by specifying their use.
\subsection{Resilience}
Below, we list the hyperparameters of the resilience experiment.
\begin{lstlisting}[language=json,firstnumber=1]
{
    "title": "Resilience of Solidago to untrustworthy users",
    "ylegend": "Correlation between learned global and ground truth",
    "xlegend": "Fraction of trustworthy users",
    "xparameter": "generative_model.user_model.p_trustworthy",
    "xvalues": [0.01, 0.2, 0.5, 0.8, 1.0], 
    "zparameter": "pipeline.aggregation.lipschitz", 
    "zvalues": [0.1, 0.3, 1, 3],
    "zlegends": ["Lipschitz = 0.1", "Lipschitz = 0.3", "Lipschitz = 1", "Lipschitz = 3"],
    "n_users": 30, 
    "n_entities": 50,
    "n_seeds": 100, 
    "generative_model": {
        "user_model": ["NormalUserModel", {
            "p_trustworthy": 0.8, 
            "p_pretrusted": 0.2, 
            "zipf_vouch": 2.0, 
            "zipf_compare": 1.5, 
            "poisson_compare": 30.0, 
            "n_comparisons_per_entity": 3.0, 
            "multiplicator_std_dev": 1.0,
            "svd_mean": [3.0, 0.0],
            "engagement_bias_std_dev": 0.0
        }], 
        "vouch_model": ["ErdosRenyiVouchModel"], 
        "entity_model": ["NormalEntityModel", {
            "mean": [0.0, 0.0]
        }], 
        "engagement_model": ["SimpleEngagementModel", {
            "p_per_criterion": {"0": 1.0}, 
            "p_private": 0.2
        }], 
        "comparison_model": ["KnaryGBT", {
            "n_options": 21, 
            "comparison_max": 10
        }]
    }, 
    "pipeline": {
        "trust_propagation": ["LipschiTrust", {
            "pretrust_value": 0.8, 
            "decay": 0.8, 
            "sink_vouch": 5.0, 
            "error": 1e-08
        }], 
        "voting_rights": ["AffineOvertrust", {
            "privacy_penalty": 0.5, 
            "min_overtrust": 2.0, 
            "overtrust_ratio": 0.1
        }], 
        "preference_learning": ["UniformGBT", {
            "prior_std_dev": 7, 
            "comparison_max": 10, 
            "convergence_error": 1e-05, 
            "cumulant_generating_function_error": 1e-05
        }], 
        "scaling": ["ScalingCompose", [
                ["Mehestan", {
                    "lipschitz": 1, 
                    "min_activity": 1, 
                    "n_scalers_max": 100, 
                    "privacy_penalty": 0.5, 
                    "user_comparison_lipschitz": 10, 
                    "p_norm_for_multiplicative_resilience": 4.0, 
                    "error": 1e-05
                }], 
                ["QuantileZeroShift", {
                    "zero_quantile": 0.15, 
                    "lipschitz": 0.1, 
                    "error": 1e-05
                }]
            ]
        ], 
        "aggregation": ["QuantileStandardizedQrMedian", {
            "dev_quantile": 0.9, 
            "lipschitz": 0.1, 
            "error": 1e-05
        }], 
        "post_process": ["Squash", {
            "score_max": 100
        }]
    }
}
\end{lstlisting}
\subsection{Engagement bias}
Below, we list the hyperparameters of the engagement-bias experiment.
\begin{lstlisting}[language=json,firstnumber=1]
{
    "title": "The role of Mehestan to fix engagement bias",
    "ylegend": "Correlation between learned global and ground truth",
    "xlegend": "Amount of engagement bias",
    "xparameter": "generative_model.user_model.engagement_bias_std_dev",
    "xvalues": [0.01, 1.0, 4.0, 10.0], 
    "zparameter": "pipeline.scaling.scalings.0.lipschitz", 
    "zvalues": [1.0, 3.0, 10.0, 30.0, 100.0],
    "zlegends": ["Lipschitz = 1", "Lipschitz = 3", "Lipschitz = 10", "Lipschitz = 30", "Lipschitz = 100"],
    "n_users": 30, 
    "n_entities": 100,
    "n_seeds": 100, 
    "generative_model": {
        "user_model": ["NormalUserModel", {
            "p_trustworthy": 0.8, 
            "p_pretrusted": 0.5, 
            "zipf_vouch": 1.5, 
            "zipf_compare": 1.5, 
            "poisson_compare": 30.0, 
            "n_comparisons_per_entity": 3.0, 
            "multiplicator_std_dev": 1.0,
            "svd_mean": [3.0, 0.0],
            "engagement_bias_std_dev": 10.0
        }], 
        "vouch_model": ["ErdosRenyiVouchModel"], 
        "entity_model": ["NormalEntityModel", {
            "mean": [0.0, 0.0]
        }], 
        "engagement_model": ["SimpleEngagementModel", {
            "p_per_criterion": {"0": 1.0}, 
            "p_private": 0.2
        }], 
        "comparison_model": ["KnaryGBT", {
            "n_options": 21, 
            "comparison_max": 10
        }]
    }, 
    "pipeline": {
        "trust_propagation": ["LipschiTrust", {
            "pretrust_value": 0.8, 
            "decay": 0.8, 
            "sink_vouch": 5.0, 
            "error": 1e-08
        }], 
        "voting_rights": ["AffineOvertrust", {
            "privacy_penalty": 0.5, 
            "min_overtrust": 2.0, 
            "overtrust_ratio": 0.1
        }], 
        "preference_learning": ["UniformGBT", {
            "prior_std_dev": 7, 
            "comparison_max": 10, 
            "convergence_error": 1e-05, 
            "cumulant_generating_function_error": 1e-05
        }], 
        "scaling": ["ScalingCompose", [
                ["Mehestan", {
                    "lipschitz": 1, 
                    "min_activity": 1, 
                    "n_scalers_max": 100, 
                    "privacy_penalty": 0.5, 
                    "user_comparison_lipschitz": 10, 
                    "p_norm_for_multiplicative_resilience": 4.0, 
                    "error": 1e-05
                }], 
                ["QuantileZeroShift", {
                    "zero_quantile": 0.15, 
                    "lipschitz": 0.1, 
                    "error": 1e-05
                }]
            ]
        ], 
        "aggregation": ["QuantileStandardizedQrMedian", {
            "dev_quantile": 0.9, 
            "lipschitz": 0.1, 
            "error": 1e-05
        }], 
        "post_process": ["Squash", {
            "score_max": 100
        }]
    }
}
\end{lstlisting}
\section{Selecting the hyperparameters}
\label{sec:section}
In this section, we discuss our hyperparameter selection,
by analyzing how the hyperparameters affect the distributions of individual and global squashed scores.
In particular, we study how the distributions of the scores vary,
depending on the number of comparisons and of contributors the alternatives received.
Recall that, for Tournesol, the selected hyperparameters are
$\contributorpriorweight \triangleq 0.02$,
$\zeroshiftquantile \triangleq 0.15$,
$\sssdevquantile \triangleq 0.9$ and
$\lipschitz \triangleq 0.1$.
The graphs are all obtained based on Tournesol's full dataset, which contains private inputs.
Note that similar graphs can be obtained using only Tournesol's public dataset, which can be downloaded from Tournesol's website at~\texttt{https://tournesol.app/\#research}.
\subsection{The prior weight $\contributorpriorweight$}
For the relatively large value $\contributorpriorweight = 0.2$ (see bottom left graph of Figure~\ref{fig:hyperparameter_alpha},
we see that alternatives can only reach very high and very low scores if they are compared a large number of times.
This is typical of a regularization that is too large.
On the other hand, our concern for $\contributorpriorweight = 0.002$ was that the noise could be too amplified,
i.e. a video may already have an extreme score despite being compared only once.
\begin{figure}[ht]
    \centering
    \includegraphics[width=.49\linewidth]{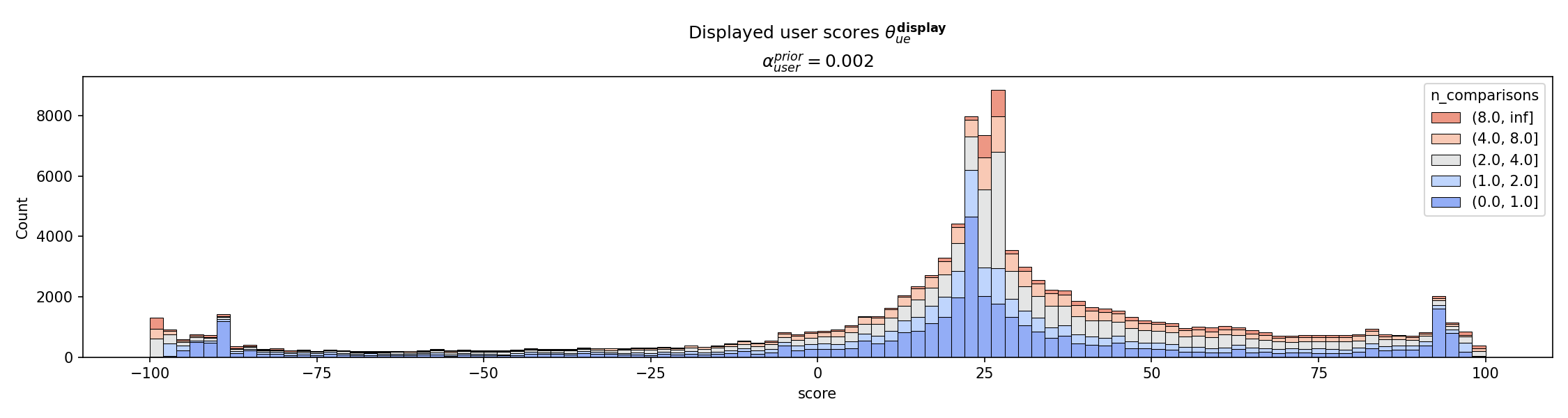} 
    \includegraphics[width=.49\linewidth]{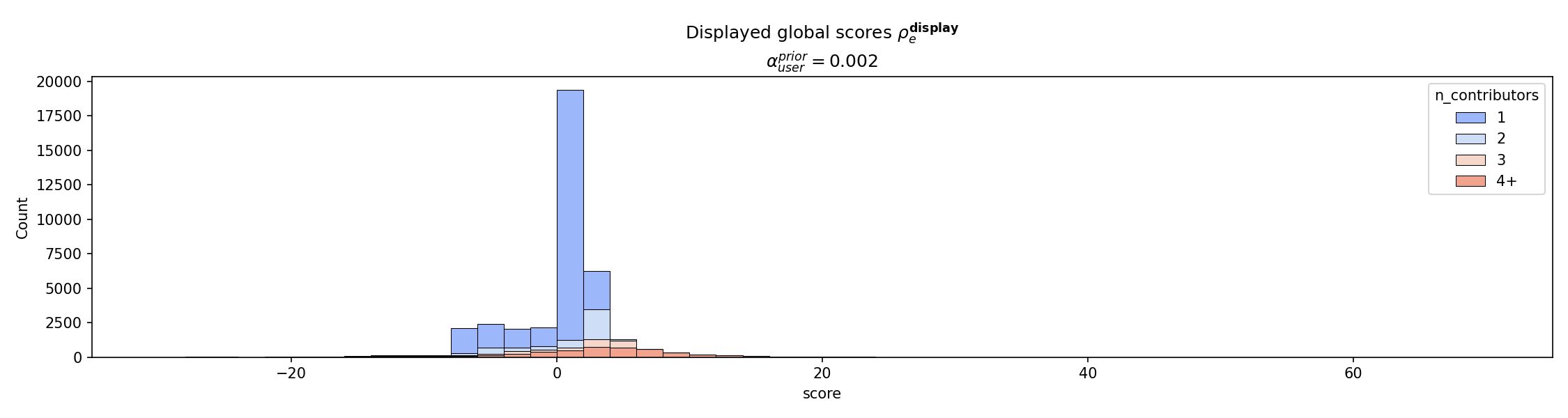} \\
    \includegraphics[width=.49\linewidth]{default_indiv_scores.png} 
    \includegraphics[width=.49\linewidth]{default_global_scores.png} \\
    \includegraphics[width=.49\linewidth]{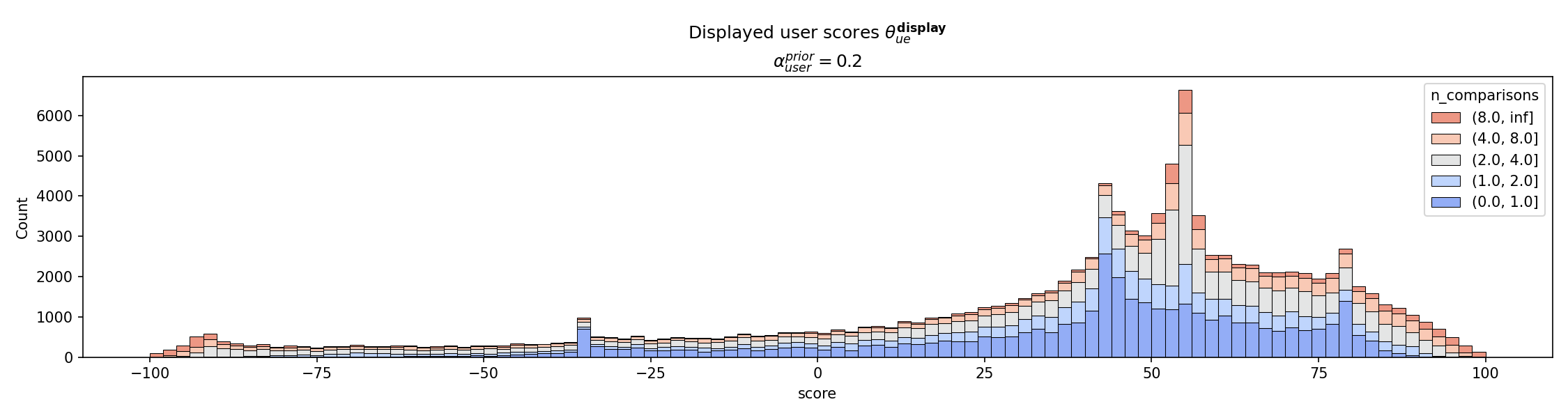} 
    \includegraphics[width=.49\linewidth]{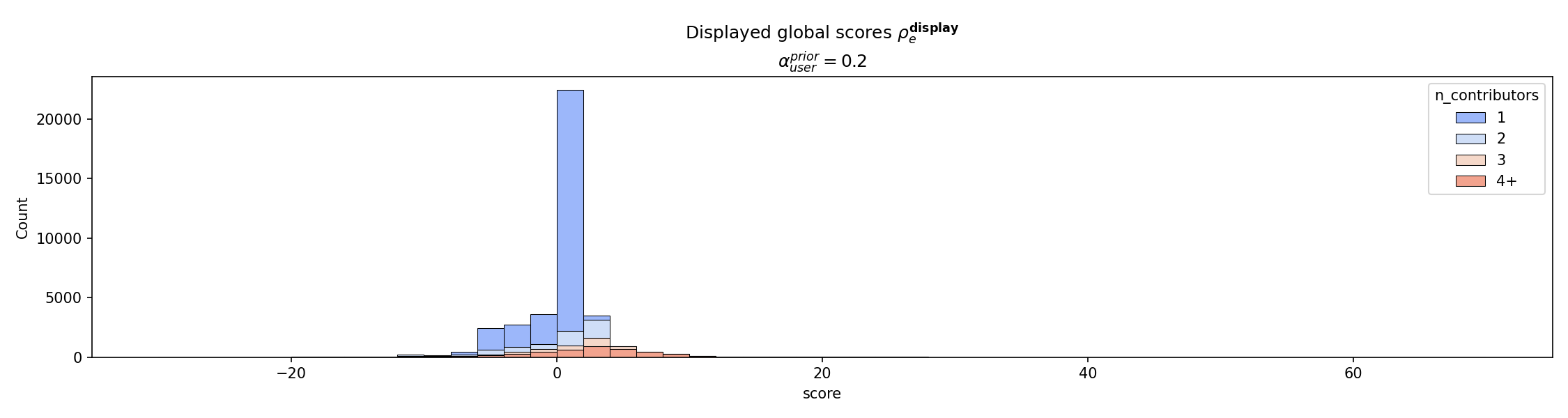}
    \caption{Impact of the prior weight $\contributorpriorweight$ on the distributions of squashed individual and global scores}
    \label{fig:hyperparameter_alpha}
\end{figure}
\subsection{The zero-shifting quantile $\zeroshiftquantile$}
The effect of the zero-shifting quantile $\zeroshiftquantile$ is clearly visible in Figure~\ref{fig:hyperparameter_q_shift}.
Namely, as the quantile $\zeroshiftquantile$ grows (from top to bottom of the figure), the distribution of individual scores is shifted closer to the middle.
We selected the middle value $\zeroshiftquantile$ as most individual scores were then around 50, 
leaving enough room on the upper end to distinguish top videos,
while also marking a clear difference between rated videos on Tournesol
and those that were never rated,
thereby enforcing a form of \emph{presumption of non-recommendability} for alternatives without evaluation.
\begin{figure}[ht]
    \centering
    \includegraphics[width=0.49\linewidth]{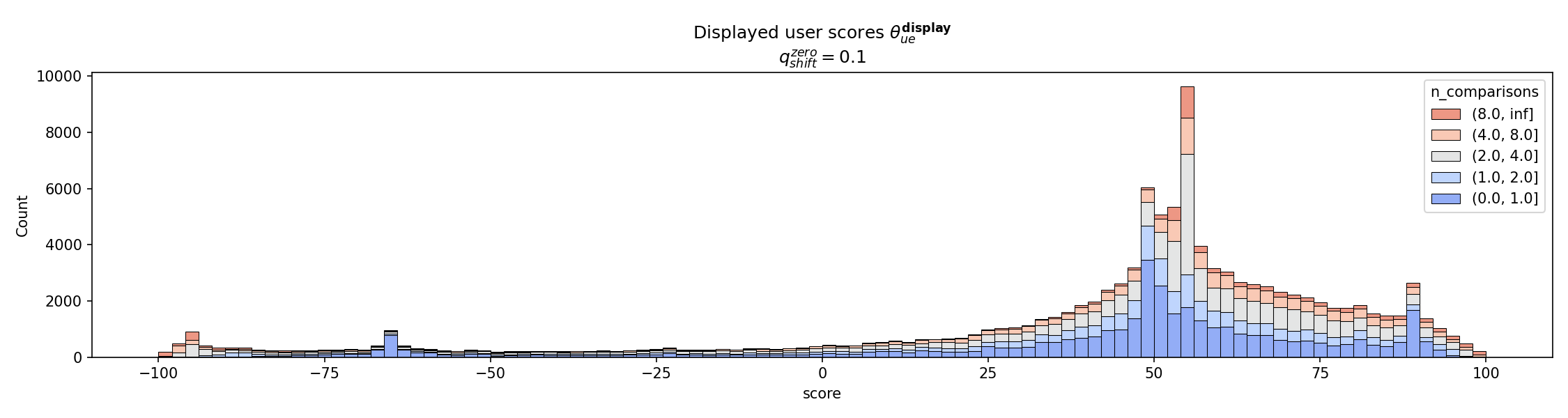} 
    \includegraphics[width=0.49\linewidth]{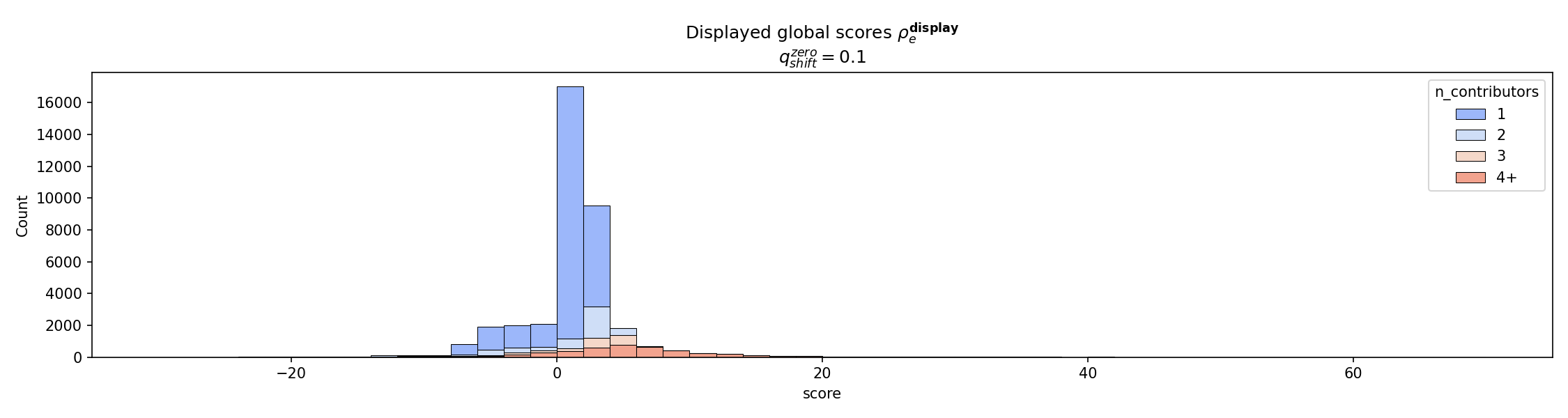} \\
    \includegraphics[width=.49\linewidth]{default_indiv_scores.png} 
    \includegraphics[width=.49\linewidth]{default_global_scores.png} \\
    \includegraphics[width=0.49\linewidth]{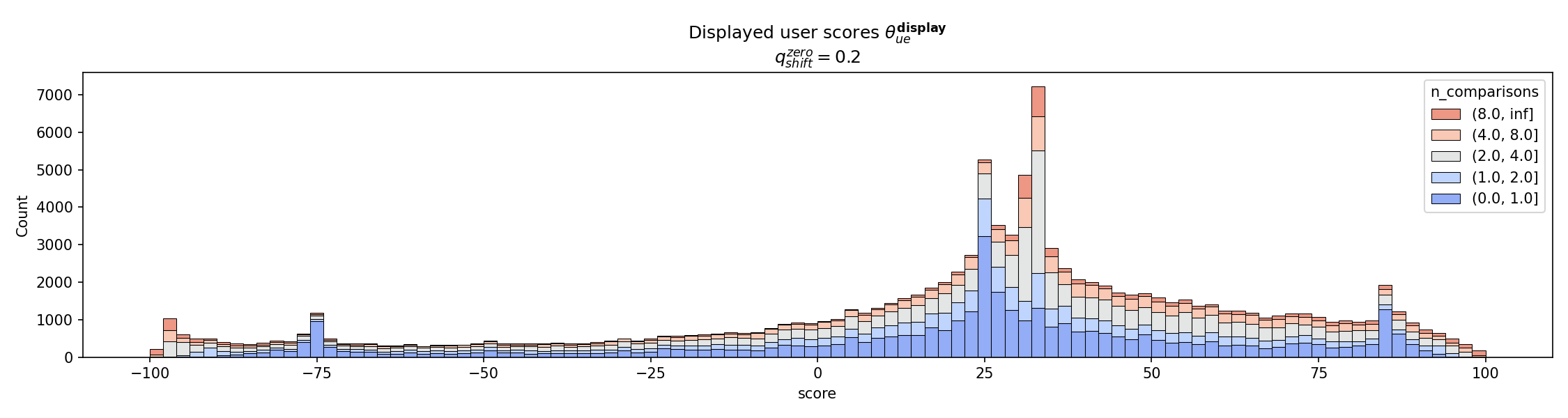} 
    \includegraphics[width=0.49\linewidth]{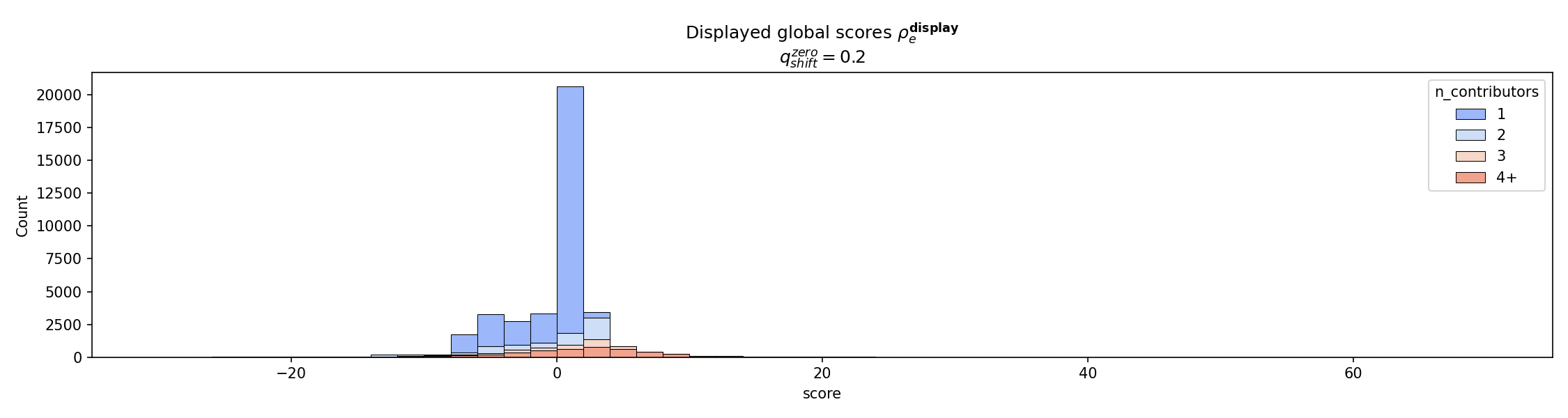} \\ 
    \caption{Impact of the zero-shifting quantile $\zeroshiftquantile$ on the distributions of squashed individual and global scores}
    \label{fig:hyperparameter_q_shift}
\end{figure}
\subsection{The collaboratively scaled individual score deviation quantile $\sssdevquantile$}
As explained the main part of the paper,
the distribution of collaboratively scaled scores is very heavy-tailed.
This is clearly evidenced by the top-left graph of Figure~\ref{fig:hyperparameter_q_dev}.
Indeed, the median deviation to the median (which corresponds to  $\sssdevquantile = 0.5$) is significantly smaller than the deviations of many scores, 
which implies that, after standardizing, shifting and squashing using  $\sssdevquantile = 0.5$,
many individual scores take the extreme values $-100$ and $100$.
We leaned on  $\sssdevquantile = 0.9$, as we deem it reasonable if at most $10\%$ of the scores take the extreme values $-100$ and $100$.
Empirically, we found that such extreme values are rarer, which makes this value reasonable.
\begin{figure}[ht]
    \centering
    \includegraphics[width=0.49\linewidth]{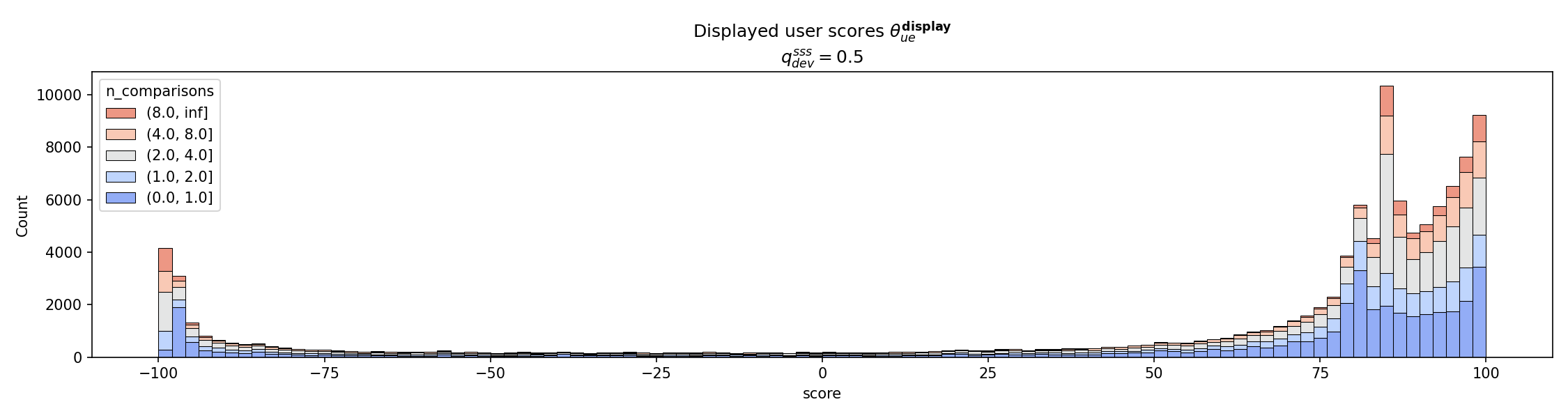} 
    \includegraphics[width=0.49\linewidth]{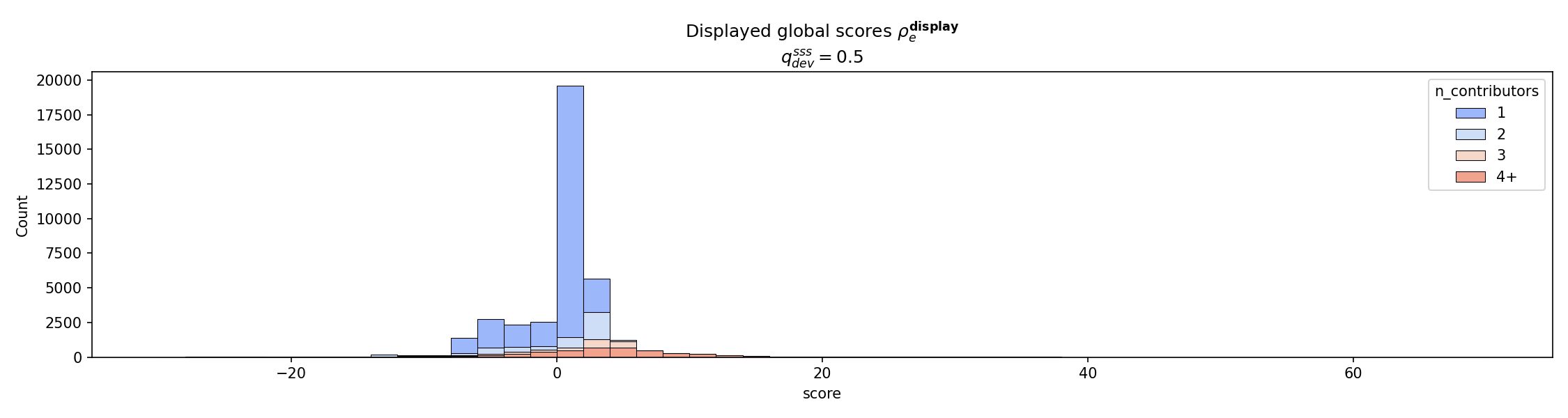} \\
    \includegraphics[width=.49\linewidth]{default_indiv_scores.png} 
    \includegraphics[width=.49\linewidth]{default_global_scores.png} \\
    \includegraphics[width=0.49\linewidth]{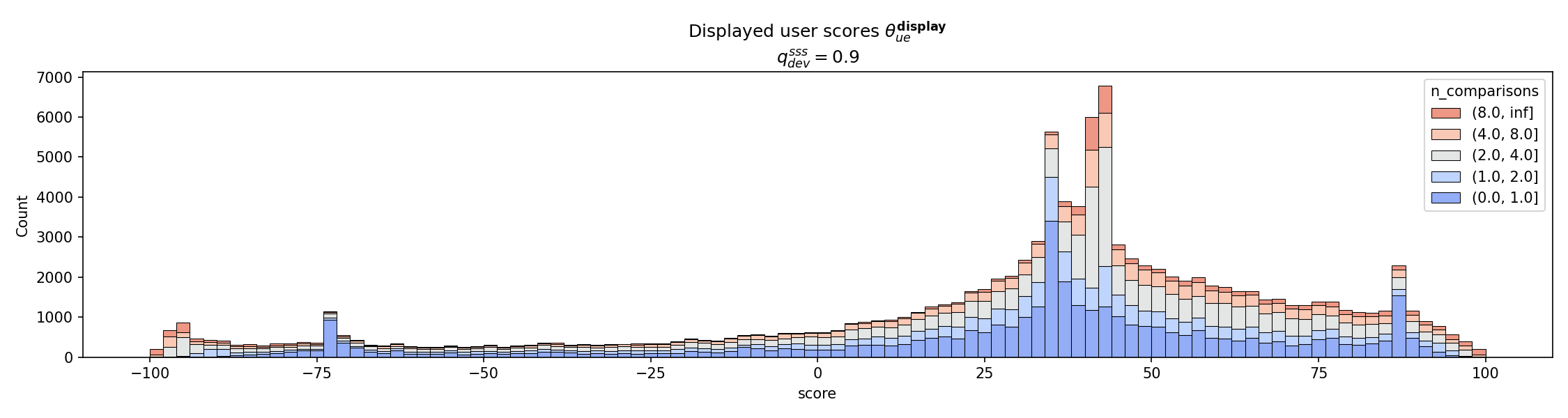} 
    \includegraphics[width=0.49\linewidth]{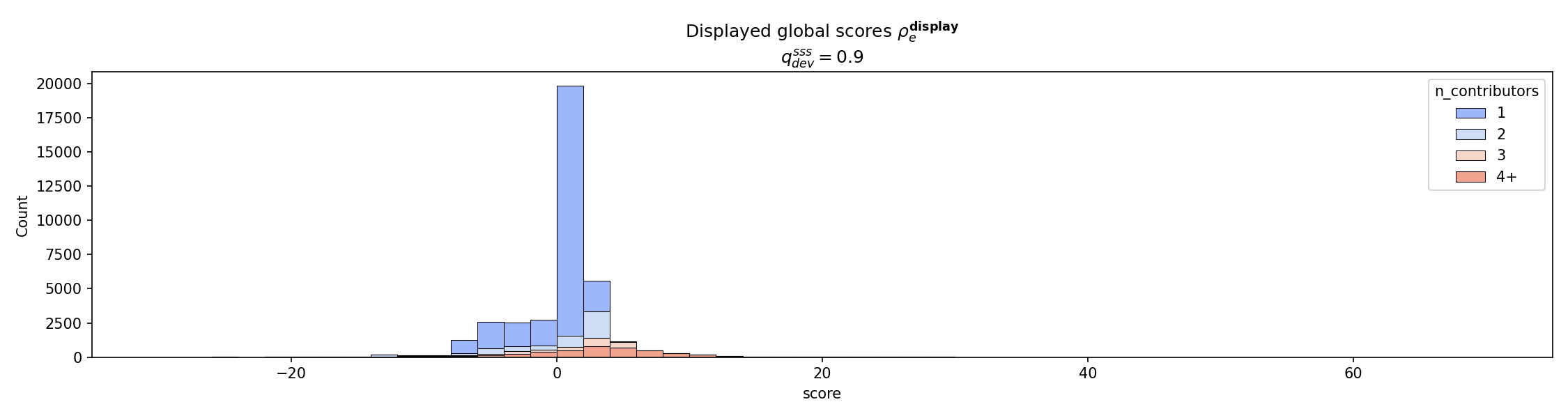} \\ 
    \caption{Impact of the collaboratively scaled individual score deviation quantile $\sssdevquantile$ on the distributions of squashed individual and global scores}
    \label{fig:hyperparameter_q_dev}
\end{figure}
\subsection{The Lipschitz resilience $\lipschitz$}
Lipschitz resilience only really matters when aggregating scores into global scores,
which is why we only plot them in Figure~\ref{fig:hyperparameter_L}.
We see that a smaller choice of $\lipschitz$ implies that videos must be scored by many contributors to have scores that really deviate from zero.
However, we feared that for $\lipschitz = 0.05$, only videos score by a very large number of users can obtain a noteworthy recommendability score.
On the other hand, for $\lipschitz = 0.2$, the global scores can already exceed 40 with only 3 contributors.
Note that Tournesol also uses a recommendability threshold, which we set at $20$.
In other words, Tournesol will not recommend any video whose global score is below this threshold.
Interestingly, given that each contributor $\contributor$ can directly influence the global scores $\globalscore$ only by $\lipschitz \votingright_\contributor$,
since the voting rights are always at most unit ($\votingright_\contributor \leq 1$),
and since the displayed scores are squashed,
for $\lipschitz = 0.1$,
the maximal score of a video with two contributors is 
$100 \frac{2 \lipschitz}{\sqrt{1 + (2 \lipschitz)^2}} \approx 19.6 < 20$.
Put differently, at least three contributors must be involved to make a content recommendable on Tournesol.
Evidently, Tournesol may modify the security hyperparameter $\lipschitz$ as the platform gains in importance.
\begin{figure}[ht]
    \centering
    \includegraphics[width=0.49\linewidth]{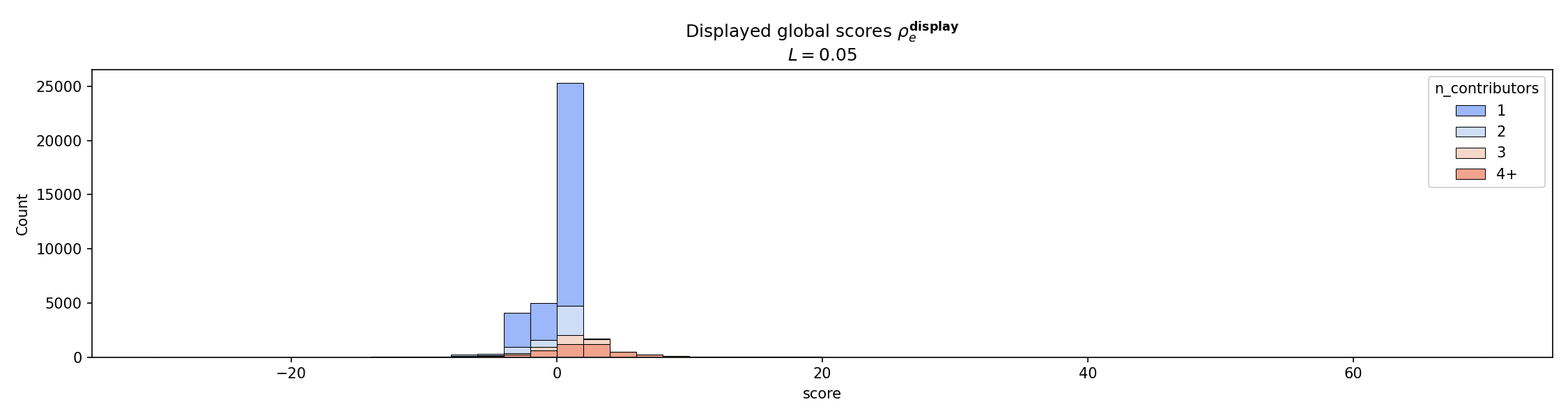} \\
    \includegraphics[width=.49\linewidth]{default_global_scores.png} \\
    \includegraphics[width=0.49\linewidth]{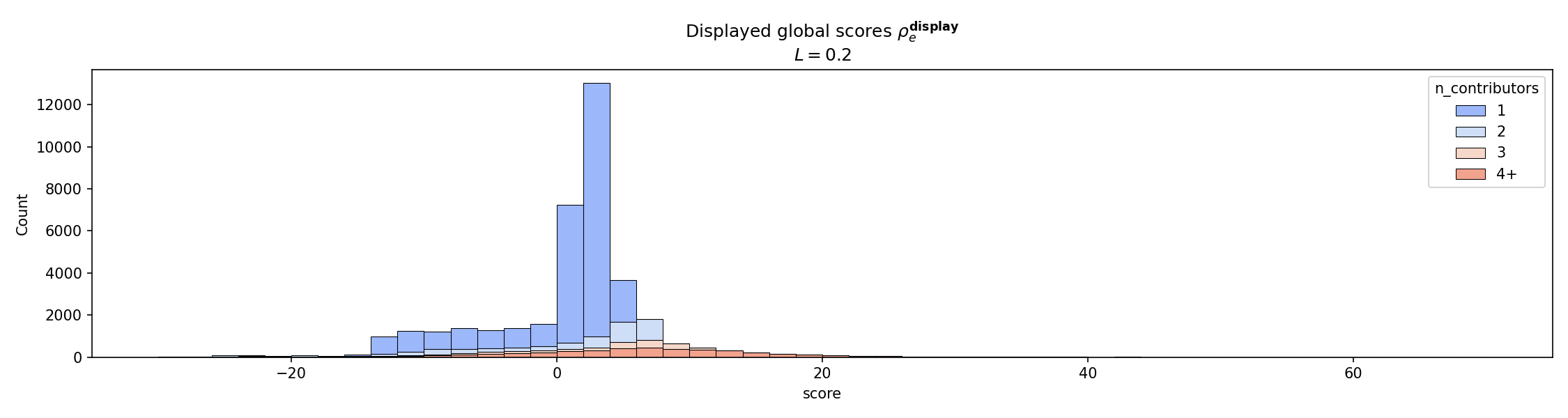} \\ 
    \caption{Impact of the Lipschitz resilience $\lipschitz$ on the distributions of squashed individual and global scores}
    \label{fig:hyperparameter_L}
\end{figure}
\section{Impact of the aggregation quantile}
\label{app:quantile_change}
This section presents experiments on the Tournesol public dataset~\cite{tournesol},
where the aggregation quantile $\alpha$ is changed, from $0.1$ to $0.8$.
Namely, we consider the default pipeline, 
but with an aggregation whose quantile changes.
Moreover, to limit the effect of low numbers of evaluations of a video,
we both consider a very large Lipschitz resilience for aggregation ($L \triangleq 100$),
and we only consider 484 videos with at least 100 comparisons from at least 10 contributors.
\begin{figure}[ht]
    \centering
    \includegraphics[width=.49\linewidth]{default_global_scores.png} \\
    \includegraphics[width=0.49\linewidth]{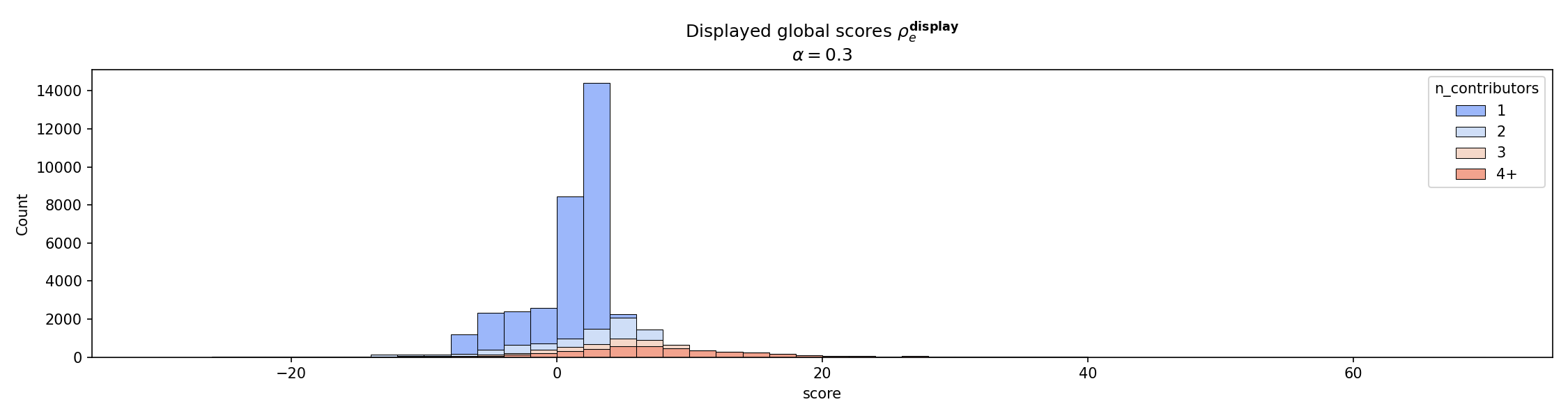} \\ 
    \includegraphics[width=0.49\linewidth]{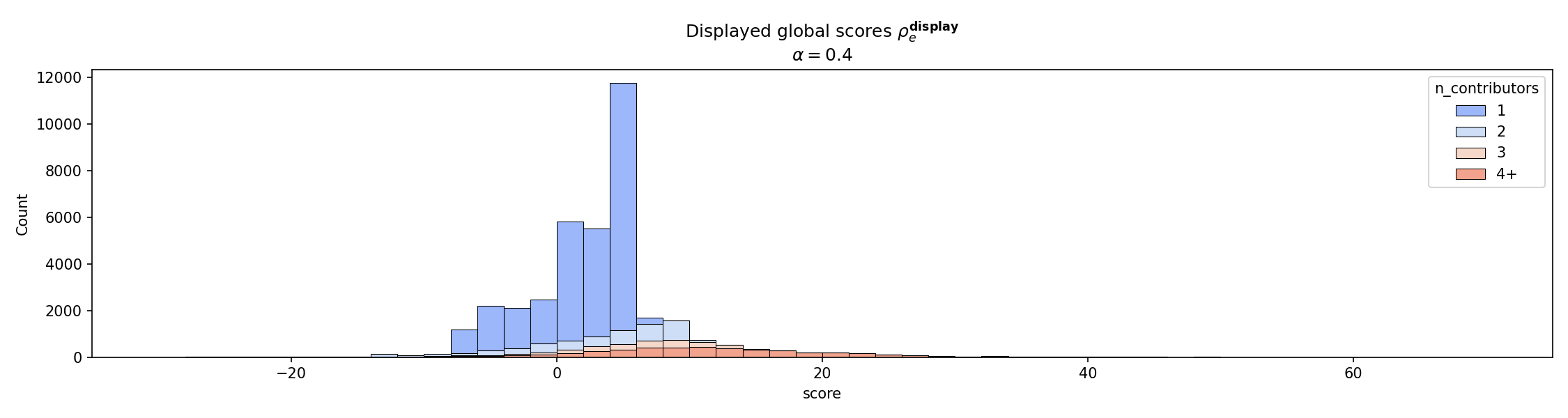} \\
    \includegraphics[width=0.49\linewidth]{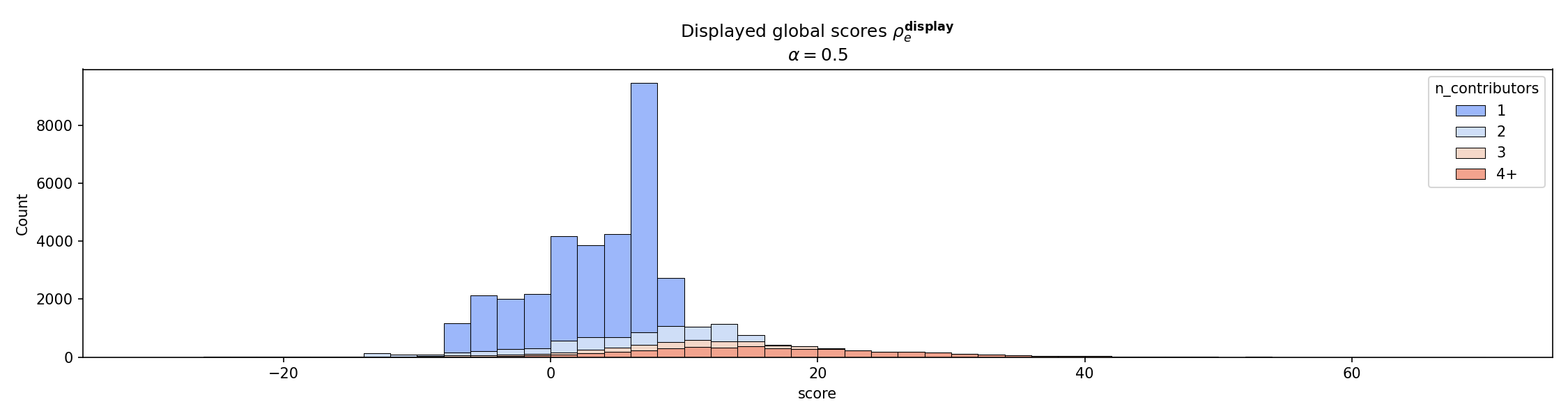} \\
    \caption{Impact of the aggregation quantile $\alpha$ on the distributions of squashed individual and global scores}
    \label{fig:hyperparameter_L}
\end{figure}
Figure~\ref{fig:quantile_change} tracks the ranking of a sample of 100 of these videos,
when $\alpha$ is changed.
\begin{figure}
    \centering
    \begin{subfigure}{0.49\linewidth}
        \includegraphics[width=\linewidth]{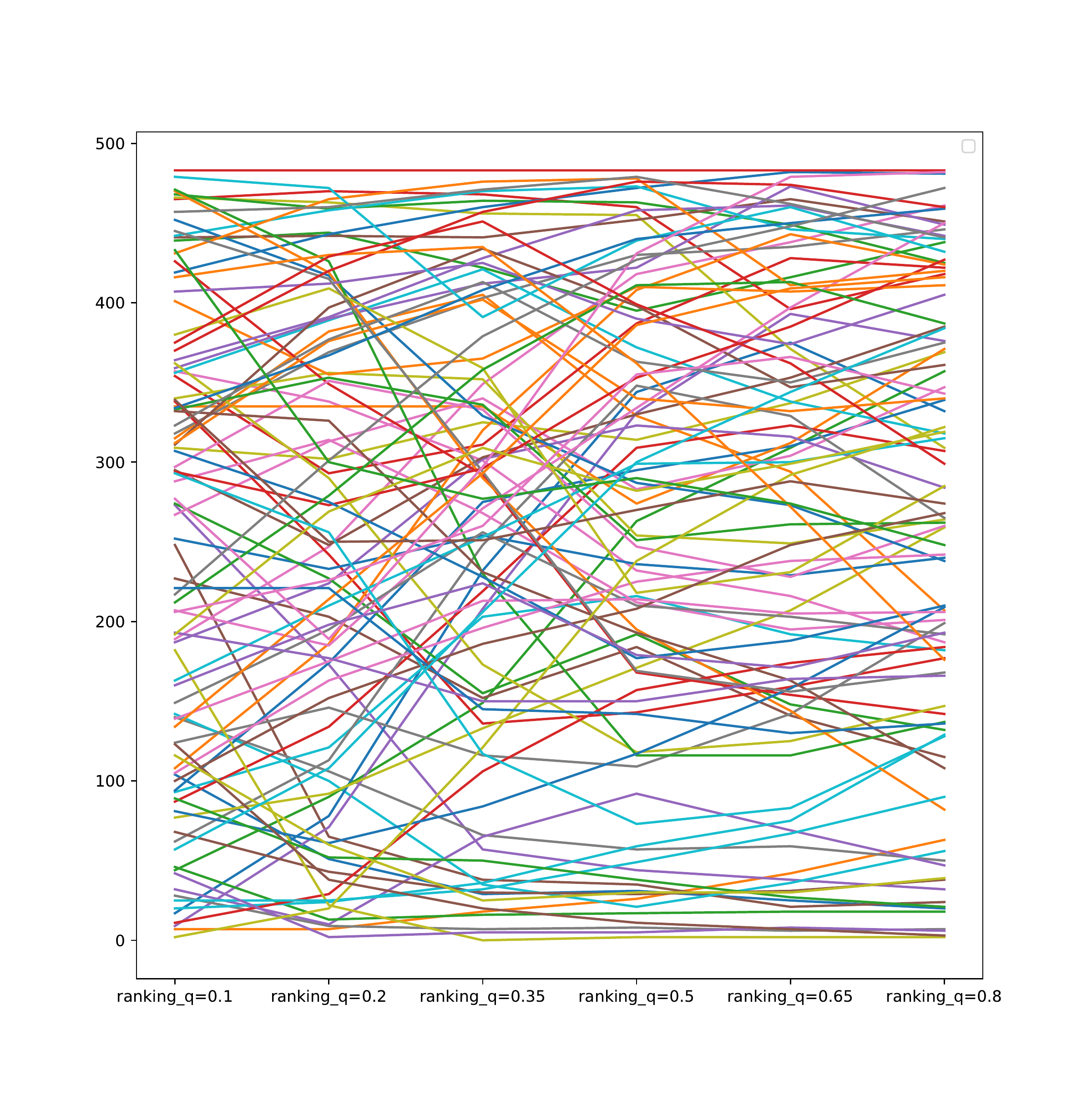}    
        \caption{Among a sample of 100 videos.}
    \end{subfigure}
    \begin{subfigure}{0.49\linewidth}
        \includegraphics[width=\linewidth]{ranking_largest_change.pdf}    
        \caption{Among the 10 videos with the largest ranking changes.}
    \end{subfigure}
    \caption{Changes in ranking}
    \label{fig:quantile_change}
\end{figure}
Perhaps unsurprisingly, the videos with the most negative ranking change are politicized,
or discuss polarizing topics such as meat consumption.
\section{Open research challenges}
\label{sec:future}
In this section, we briefly list research venues which we believe to be important for any secure and collaborative algorithmic governance system design.
\subsection{Accounting for direct assessments}
Users issue preferences between alternatives, which are converted into individual scores via the generalized Bradley-Terry model. We also want to make it possible to express assessments (likes or dislikes on the alternatives) which will be taken into account in the estimation of the scores. This presupposes adapting the BT type models to these new data, what we plan for future work.
\subsection{Distinguishing instinctive preferences from thoughtful volitions}
It has been empirically shown that more reflexive judgments, 
which are related to what philosophers sometimes call volitions~\cite{adams1992intention, zhu2004understanding},
lead to statistically different judgments.
Some studies even find that they seem more inclusive to outgroup members~\cite{agan2023automating}.
As initiated by~\cite{DBLP:conf/hci/LechiakhM23}, it would be worthwhile to construct algorithms
capable of making such a distinction, by learning systematic bias between judgments provided with more or less thought.
In particular, the Tournesol dataset contains both judgments made only along the main criterion,
as well as judgments made while considering all criteria, which area likely to be more thoughtful.
Future work should investigate whether this can be leveraged.
\subsection{Enabling liquid democracy}
One of the key limitations of the platform, is that most humans are not very active on it.
To increase (indirect) participation, liquid democracy was proposed~\cite{carroll1884principles, ramos2015liquid, DBLP:conf/sigecom/0002HJMPR23},
where each voter can delegate their vote to any other voter of their choice.
A platform like Tournesol seems like a promising venue to further develop, deploy and test such principles,
and to adjust them to security and undesirable group dynamic concerns.
\subsection{Leveraging (vouched) expertise}
Future work should investigate how to certify and leverage contributors' expertise, 
e.g. based on the topic of a content and the quality criterion under consideration.
In particular, it would be worth partnering with institutions like Amnesty or the IPCC, e.g.,
to know their members' top content recommendations.
A peer-to-peer and institution-to-peer vouching system may be a promising approach for certification.
\subsection{Online updates}
Our algorithms are unfortunately slow to run.
On Tournesol, they require 40 minutes of computation.
This means that they cannot be run whenever a contributors enters a comparison,
which prevents near-instantaneous feedbacks.
Heuristics have yet to be designed to allow this, 
which is arguably a key component for interpretability.
Moreover, the guarantees of such heuristics would need to be analyzed.
\subsection{Distributed computing}
Future work should also aim to distribute the computational workload,
e.g. by asking contributors to perform operations on their machines.
Moreover, distributed computing could allow distributed verifiability,
i.e. contributors would no longer have to trust that a central server (like Tournesol's) correctly runs our algorithms.
\subsection{Active learning}
Tournesol mostly relies on contributors to select which alternatives to compare.
Future work should investigate active learning algorithms to best help contributors to effortlessly provide highly information-valuable comparisons.
Keeping contributors motivated is particularly important and challenging.
\subsection{Robust raw score generalization.}
For security reasons, contributors currently only influence the score of the alternatives they explicitly compared.
Using machine learning, contributors' scores could be generalized,
to essentially guess how a given contributor would have likely scored the alternatives that they did not explicitly compare,
and to take this into account to compute global scores.
\subsection{Byzantine collaborative filtering.}
Another approach to score generalize relies on \emph{collaborative filtering}.
However, collaborative filtering has little security guarantees.
Future research should investigate the design of Byzantine-resilient collaborative filtering algorithms.
\subsection{Bayesian voting.}
Unfortunately, we suspect large biases in the population of contributors.
To correct these biases, a promising approach would consist of predicting 
what non-participating humans would have likely voted,
e.g. based on what similar contributors have voted.
Intuitively, this would amount to assigning more voting rights to contributors
who represent large under-represented populations.
This raises serious security concerns though,
as malicious contributors would be incentivize to claim to be from such under-represented populations 
to increase their voting rights.
\subsection{Measuring Tournesol's impact.}
Since Tournesol's end goal is to improve information prioritization,
it is important to determine the extent to which 
the daily use of Tournesol modifies users' content consumption habits.
Experiments should be conducted to do so.
\subsection{Aligned language models.}
Tournesol essentially collaboratively labels which content is desirable to repeat.
This could be used to train language models that are more aligned with human preferences.
\subsection{Provide end-to-end guarantees}
We did not provide end-to-end security guarantees.
Doing so without sacrificing accuracy is another of the many exciting future works.
But as a result, we expect Tournesol's algorithms to be regularly updated in the coming years.
Nevertheless, we believe that much of the pipeline structure presented in the paper will remain relevant,
both for Tournesol and other large-scale secure participatory projects.

\end{document}